\documentclass[a4paper,fleqn]{cas-sc}

\usepackage[numbers,sort&compress]{natbib}
\usepackage{amsmath,amssymb,amsthm,mathtools}
\usepackage{booktabs}
\usepackage{enumitem}
\usepackage{graphicx}
\usepackage{xcolor}
\usepackage{hyperref}
\usepackage{capt-of}
\usepackage{multirow}
\usepackage{wrapfig}
\usepackage{float}
\usepackage{placeins}

\newtheorem{theorem}{Theorem}
\newtheorem{corollary}{Corollary}
\newtheorem{proposition}{Proposition}
\newtheorem{definition}{Definition}

\newcounter{algorithm}
\newenvironment{algorithmblock}[1]{%
  \refstepcounter{algorithm}%
  \par\smallskip\noindent\hrule\smallskip%
  \textbf{Algorithm~\thealgorithm. #1}\par\smallskip%
  \begingroup\small%
}{%
  \par\endgroup\smallskip\hrule\par\smallskip%
}

\newcommand{\R}{\mathcal R}
\newcommand{\EO}{\operatorname{EO}}
\newcommand{\SE}{\operatorname{SE}}

\newcommand{\eps}{\varepsilon}

\begin{document}
\let\WriteBookmarks\relax
\def\floatpagepagefraction{1}
\def\textpagefraction{.001}

\shorttitle{EP-realizability for Dynamics-preserving and Privacy-aware Network Reconstruction}
\shortauthors{R. Porcedda}

\title[mode=title]{Equitable Partition Realizability for Dynamics-preserving and Privacy-aware Network Reconstruction}

\author[1,2]{Riccardo Porcedda}[orcid=0000-0001-7360-4401]
\cormark[1]
\ead{riccardo.porcedda@santannapisa.it}
\credit{Conceptualization, Methodology, Software, Writing}

\affiliation[1]{organization={Department of Excellence L'EMbeDS, Sant'Anna School of Advanced Studies},
                city={Pisa},
                country={Italy}}
\affiliation[2]{organization={Department of Computer Science, University of Pisa},
                city={Pisa},
                country={Italy}}
\cortext[cor1]{Corresponding author.}

\begin{abstract}
Degree-sequence realizability is the combinatorial basis of configuration models, but degree constraints alone do not ensure the preservation of graph dynamics. Hence, configuration models are unable to recover centrality measures, unless these are strongly correlated with the degree sequence. To address this matter, we introduce \emph{EP-realizability}, the analogue problem induced by an equitable partition (EP): given the EP of a graph, decide whether the partition is realized by a simple undirected loopless graph and therefore construct such a graph. After defining the problem, we solve it by reducing it to sub-problems related to Havel--Hakimi and the Gale--Ryser theorem. We also face the challenge of solving the problem with an Approximate Equitable Partition ($\varepsilon$-EP), so that it is possible to reconstruct a network starting from partial and more privacy-preserving information. We evaluate privacy with edge overlap, deriving also, for our proposed $\varepsilon$-EP-realizability solution, a predictor for this metric. Experiments on Karate, Cora, CiteSeer and PubMed datasets show that our algorithm achieves a favourable and tunable privacy--utility trade-off, comparing the results with Havel--Hakimi algorithm, Newman's configuration model and a stochastic block model. Finally, both with real data and random graphs, we show that our algorithm has approximately linear time complexity with respect to the number of edges.
\end{abstract}

\begin{keywords}
Equitable partition \sep Weisfeiler--Lehman \sep Graph realization \sep Configuration model \sep Havel--Hakimi \sep Privacy \sep Edge overlap \sep Network centrality
\end{keywords}

\maketitle

\section{Introduction}
\label{sec:intro}

The classical graph realization problem asks whether a prescribed degree sequence is induced by a simple graph. The Erd\H{o}s--Gallai criterion and the constructive Havel--Hakimi algorithm answer this question exactly \cite{Havel1955,ErdosGallai1960,Hakimi1962}. Together with the bipartite Gale--Ryser theorem \cite{Gale1957,Ryser1957}, these results are the combinatorial substrate of configuration models: once a degree sequence is graphical, it is possible to sample graphs from the distribution of all realizable fibres \cite{NewmanStrogatzWatts2001,Fosdick2018}.
However, the limitation of configuration models is that degree constraints alone are not designed to preserve mesoscopic structure or dynamic observables induced by walks, such as Katz scores, PageRank scores, eigenvector centrality, or entropy-based centralities \cite{cimini2021reconstructing}.

Several extensions address this limitation by adding structural constraints: correlated configuration models preserve degree--degree mixing, while layered configuration models use centrality or onion layers to constrain the ensemble \cite{HebertDufresne2016,HebertDufresne2024}. Still, these constraints do not provide algebraic guarantees on single realizations, but rather empirical results on sampled ensembles and related statistics.

To tackle this issue, we study a setting based on equitable partitions (EPs), which are closely related to stable 1-WL colour refinement and structural roles \cite{WeisfeilerLeman1968,GodsilRoyle2001}. If $H$ is the membership matrix of an exact EP (grouping nodes into structural roles) of an adjacency matrix $A$, this satisfies
\begin{equation} \label{eq:EP}
        AH=HQ
\end{equation}
with $Q$ being the adjacency matrix of the quotient graph. Every eigenpair $(\lambda,\nu)$ of
$Q$ lifts to an eigenpair of $A$ through the block-constant vector $H\nu$:
\begin{equation*}
AH\nu=HQ\nu=\lambda H\nu.
\end{equation*}
Thus the quotient spectrum is embedded in the spectrum of $A$ on the
block-constant subspace. Moreover, by induction,
\begin{equation*}
    A^tH=HQ^t
    \qquad t\geq 0,
\label{eq:EP-powers}
\end{equation*}
so any linear diffusion started from block-constant features evolves exactly
as its quotient counterpart. This result was shown in \cite{GodsilRoyle2001,cardelli2015forward}.
Other quantities predetermined by the EP and the quotient graph are eigenvector centrality \cite{scholkemper2023optimization, sanchez2020exploiting, tognazzi2018differential}, PageRank centrality \cite{scholkemper2023optimization,tognazzi2018differential} and Katz centrality \cite{allouch2025key, tognazzi2018differential}.
Examples of EP adoption for non-linear dynamics like consensus and epidemic dynamics can also be found in \cite{o2013observability, bonaccorsi2015epidemic}.

For many applications, it may also be desirable to have approximate definitions of EP ($A H \simeq HQ$) to further compress aggregate information while preserving dynamics as much as possible \cite{bonaccorsi2015epidemic, scholkemper2023optimization, squillace2024efficient, squillace2026scalable}.

Therefore, we explore in this paper if the concept of EP (and approximate EP) can be used to define a new graph realization problem, whose solution can be the combinatorial substrate for new network reconstruction methods capable of preserving dynamic properties of graphs.

A related direction is the NeST
configuration model \cite{Stamm2023}. NeST samples from the fully available
original graph by switching edges while preserving WL colourings up to a chosen
depth. In contrast, our setting starts from a released aggregate role
description, such as $(AH,H)$ or $(H,Q)$, and asks what graphs can be built
from that information alone.

\subsection{Contributions}
The contributions of the paper are therefore the following:
\begin{enumerate}[leftmargin=18pt]
\item we define the (approximate) EP-realizability problem;
\item we provide a solution to the problem, proving that exact EP-realizability decomposes into independent classical
      graphicality problems;
\item from the solution, we naturally build a constructive heap-based realization algorithm with $O(n+m)\log n$ time complexity (where $n$ is the number of vertices in the graph and $m$ the number of edges);
\item we derive equations to predict, having as input $(AH,H)$ or $(H,Q)$ alone, mean and standard deviation of the edge overlap between original graphs and graphs reconstructed with our method;
\item we evaluate the edge overlap and the utility of reconstructed graphs in terms of static properties (power-law exponent, number of triangles, clustering coefficient, ...), centrality measures (eigenvector, PageRank, Katz, betweenness) and diffusion error;
\item we compare results from point 5 with graphs reconstructed using Havel--Hakimi, Newman's configuration model and a stochastic block model built on $H$;
\item the evaluation of points 3-6 are made on real-world datasets, but for point 3 we adopt also random graphs.
\end{enumerate}

\section{Notation and Problem Definition}
\label{sec:notation}

Let $G=(V,E)$ be a simple undirected graph with adjacency matrix
$A=A^\top\in\{0,1\}^{n\times n}$ and $A_{ii}=0$. Let
$\{B_1,\ldots,B_k\}$ be a partition of $V$, with
$s_\alpha=|B_\alpha|$, and let $H\in\{0,1\}^{n\times k}$ be its
membership matrix. For a graph $A$, its block-degree profile with respect to $H$ is $AH$.
Thus $(AH)_{i\beta}$ is the number of neighbours that node $i$ has in the role $B_\beta$.

\begin{definition}[$\varepsilon$-EP realization problem]
    Given an $\varepsilon$-equitable partition with membership matrix $H$ and either $AH$ or $Q=(H^T H)^{-1} H^T A H$, the $\varepsilon$-EP realization problem asks whether there is a simple graph $\hat A$ such that $\|\hat A H - H(H^T H)^{-1} H^T A H\|_\infty = \|\hat A H - HQ\|_\infty \leq \varepsilon$.
\end{definition}

Having defined the $\varepsilon$-EP realization problem, we proceed by defining the realization with respect to an arbitrary integer block-degree profile.

\begin{definition}[Block-degree realization set]
Given a partition with membership matrix $H$ and an integer matrix
$T\in\mathbb Z_{\geq0}^{n\times k}$, the block-degree realization set is
\begin{equation}
\R(T,H)
=
\{
\hat A:
\hat A=\hat A^\top,\ 
\hat A_{ii}=0,\ 
\hat A_{ij}\in\{0,1\},\ 
\hat AH=T
\}.
\end{equation}
\end{definition}

\noindent The notation $\R(T,H)$ is general enough to cover both $(AH,H)$ and $(H,Q)$ releases, for which we have $T=AH$ and $T=HQ$ respectively.

Now we can proceed with the conditions for the block-degree and EP-realizability.

\begin{figure}
    \centering
    \includegraphics[width=\linewidth]{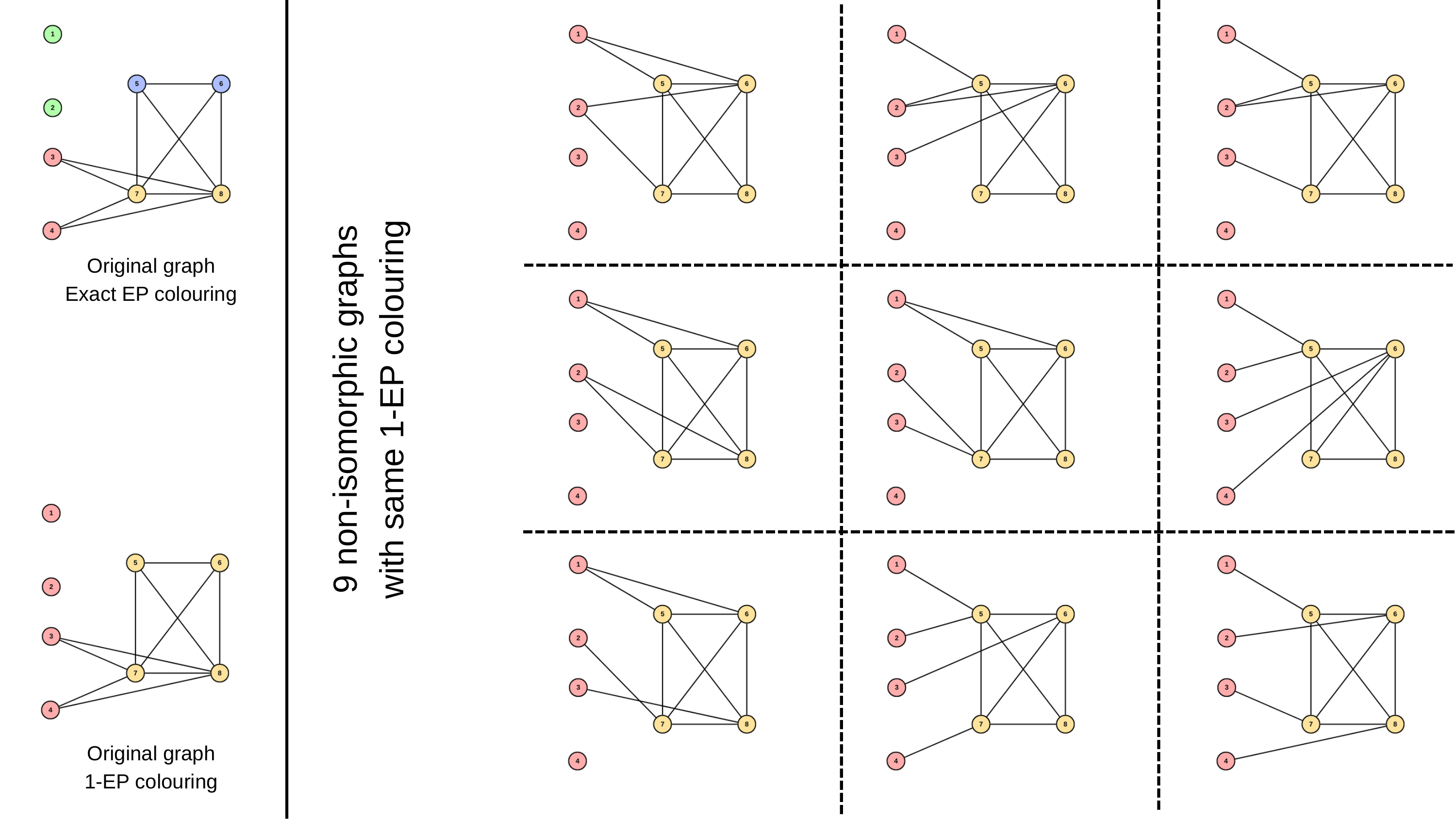} 
    \caption{\textbf{Left:} Original graph with its exact EP colouring (top) and $1$-EP colouring (bottom). With the exact EP, there is only one possible realization (the original graph itself), up to relabeling. \textbf{Right:} the other nine non-isomorphic realizations preserving the same $1$-EP representation. Together with the original graph, these form the 10 non-isomorphic realization classes corresponding to the 1428 labeled graphs compatible with the same $1$-EP.}
    \label{fig:non-isomorphic}
\end{figure}

\section{Block-degree and exact EP-realizability}
\label{sec:realizability}

We first state a general realizability criterion with block-degree $T$, then, for EP-realizability, we discuss the case $T=HQ$.

\begin{theorem}[Block-degree realizability]
\label{thm:exact}
Let $T\in\mathbb Z_{\ge0}^{n\times k}$. Then
$\mathcal R(T,H)\neq\varnothing$ if and only if:
\begin{enumerate}[label=(\roman*),leftmargin=18pt]
\item for every $\alpha$, the sequence
$(T_{i\alpha})_{i\in B_\alpha}$ is graphical;
\item for every $\alpha<\beta$, the pair
$\left(
    (T_{i\beta})_{i\in B_\alpha},
    (T_{j\alpha})_{j\in B_\beta}
    \right)$ is bipartite-graphical.
\end{enumerate}
\end{theorem}

\begin{proof}
Necessity follows by restricting any feasible $\hat A$ to its diagonal
and off-diagonal blocks. The diagonal block
$\hat A_{\alpha\alpha}$ is a simple graph on $B_\alpha$ with degree
sequence $(T_{i\alpha})_{i\in B_\alpha}$. The off-diagonal block
$\hat A_{\alpha\beta}$ is a simple bipartite graph between
$B_\alpha$ and $B_\beta$, with left and right degree sequences $(T_{i\beta})_{i\in B_\alpha}$ and
$(T_{j\alpha})_{j\in B_\beta}$, respectively.
Conversely, realize each diagonal sequence by Havel--Hakimi and each
off-diagonal pair by bipartite Havel--Hakimi, or equivalently by the
Gale--Ryser theorem. Placing these realized blocks independently and
setting $\hat A_{\beta\alpha}=\hat A_{\alpha\beta}^{\top}$ gives a symmetric loopless binary matrix satisfying $\hat A H=T$.
Therefore $\mathcal R(T,H)\neq\varnothing$.
\end{proof}


\begin{corollary}[Realizability of the block-degree release]
\label{cor:AH-release}
Let $A\in\{0,1\}^{n\times n}$ be the adjacency matrix of a simple
undirected loopless graph, and let $H$ be the membership matrix any partition. Then $\mathcal R(AH,H)\neq\varnothing$ and $T=AH$ always satisfies the
conditions of Theorem~\ref{thm:exact}.
\end{corollary}

\begin{proof}
Take $T=AH$. Then a graph with adjacency $A$ trivially satisfies the condition, hence $\mathcal R(AH,H)\neq\varnothing$.
The conditions of Theorem~\ref{thm:exact} follow immediately, since they
are necessary for any realizable target $T$.
\end{proof}

\begin{corollary}[EP-realizability with $T=HQ$]
\label{cor:HQ-release}
Let $A\in\{0,1\}^{n\times n}$ be the adjacency matrix of a simple
undirected loopless graph. Suppose that $H$ is the membership matrix of an equitable partition
of $A$, with quotient matrix $Q$.
Then $\mathcal R(HQ,H)\neq\varnothing$ and $T=HQ$ satisfies the conditions of Theorem~\ref{thm:exact}.
\end{corollary}

\begin{proof}
Since $H$ is the membership matrix of an equitable partition of $A$ with quotient $Q$, we have $AH=HQ$. Therefore, Corollary \ref{cor:AH-release} trivially follows.
\end{proof}

\section{$\varepsilon$-EP-realizability}
So far, we have discussed the exact EP-realizability, i.e. the solution of the $\varepsilon$-EP realization problem in the case $\varepsilon=0$.
Now, we want to explore the realizability conditions for $\varepsilon>0$.

The reason why we are interested in the realizations of approximate EPs is that we aim at increasing the cardinality of the realization set $\mathcal R(T,H)$: while preserving the dynamic properties to some degree, the more graphs can be realized with the same $\varepsilon$-EP, the more difficult is gonna be for an attacker to infer the original graph that produced that $\varepsilon$-EP.
To better visualize this, in Figure \ref{fig:non-isomorphic} we provide an example of a graph which has only one possible realization with exact EP (up to relabeling of the nodes), while the 1-EP increases the set to a total of 10 non-isomorphic graphs.
In Section B of the Supplementary Material we give further details about this example.

\subsection{The problem of the approximate quotient}
The solution to the $\varepsilon$-EP realization problem is actually trivial for the release $T=AH$: even with $\varepsilon>0$, the definition of the membership matrix $H$ remains unchanged, therefore $T$ is still an integer matrix, hence Theorem \ref{thm:exact} and Corollary \ref{cor:AH-release} are satisfied.
The quotient matrix $Q=(H^\top H)^{-1}H^\top AH$, however, is not integer in general when $\|A H - HQ\|_\infty \leq \varepsilon$ holds.
So, for $\varepsilon>0$, it is necessary to find an integer feasible matrix $T=\hat AH$ satisfying $\|T-HQ\|_\infty\leq\varepsilon$.
Equivalently, for every $i\in B_\alpha$ and
every $\beta$, the admissible values of $T_{i\beta}$ are
\begin{equation} \label{eq:interval}
I_{i\beta}^{\varepsilon}(Q)
=
\left[
\max\{0,\lceil Q_{\alpha\beta}-\varepsilon\rceil\},
\min\{N_{i\beta},\lfloor Q_{\alpha\beta}+\varepsilon\rfloor\}
\right]\cap \mathbb{Z},
\end{equation}
with 
$
N_{i\beta}
=
\begin{cases}
s_\beta-1, & \beta=\alpha,\\
s_\beta, & \beta\neq\alpha
\end{cases}
$
being the maximum possible number of neighbours of $i$ in
$B_\beta$ for $i\in B_\alpha$.

\begin{theorem}[Approximate quotient realizability]
\label{thm:approx}
Let $Q\in\mathbb R^{k\times k}$ and let
$\varepsilon\geq0$. There exists a simple undirected loopless graph
$\hat A$ such that
$
    \|\hat AH-HQ\|_\infty\leq\varepsilon
$
if and only if there exists an integer matrix
$T\in\mathbb Z_{\ge0}^{n\times k}$ such that
$
    T_{i\beta}\in I_{i\beta}^{\varepsilon}(Q)
    \quad \forall i,\beta,
$
with $I_{i\beta}^{\varepsilon}(Q)$ as defined in Eq. \ref{eq:interval}. This $T$ satisfies the block-wise graphicality conditions of
Theorem~\ref{thm:exact}.
\end{theorem}

\begin{proof}
If such a graph $\hat A$ exists, set $T=\hat AH$. Then $T$ has
nonnegative integer entries, lies in the intervals
$I_{i\beta}^{\varepsilon}(Q)$, and is block-wise graphical by
Theorem~\ref{thm:exact}. Conversely, if such a $T$ exists,
Theorem~\ref{thm:exact} gives a graph $\hat A$ with $\hat AH=T$. Since
$T_{i\beta}\in I_{i\beta}^{\varepsilon}(Q)$, we obtain
$
    \|\hat AH-HQ\|_\infty
    =
    \|T-HQ\|_\infty
    \leq \varepsilon.
$
\end{proof}

\subsection{$\varepsilon$-EP realization algorithm}
\label{sec:algorithm}

Theorem~\ref{thm:exact} provides a constructive algorithm (Algorithm \ref{alg:ep-realization}) both for exact and approximate EP for a node-level release $(AH,H)$, while only for exact EP in the case of a quotient-level release $(H,Q)$. If the approximate EP collapses to a single block in the case of the release $(AH,H)$, then the method coincides with Havel--Hakimi when $b=\theta=0$, or with Newman's configuration model otherwise. For a quotient-level release $(H,Q)$ with $\varepsilon>0$, the algorithm first chooses an integer feasible profile $T=\Phi_\eps(H,Q,\varepsilon)$ according to Theorem~\ref{thm:approx} and then Theorem~\ref{thm:exact} applies.

\begin{algorithmblock}{$\varepsilon$-EP realization}
\label{alg:ep-realization}
\textbf{Input:} partition \(\{B_1,\ldots,B_k\}\) with membership matrix $H$; release
type $(AH,H)$ or $(H,Q)$; tolerance \(\varepsilon\); sampling
parameters \(S,b,\theta\).

\textbf{Output:} one or more simple loopless realizations \(\hat A\), or None if infeasible.

\begin{enumerate}[leftmargin=18pt]
\item If \(\rho=(AH,H)\), set $T=AH$.

\item If \(\rho=(H,Q)\), $T= \Phi_\varepsilon(H,Q)$.
If the selector returns None, return None.

\item Test the block-wise graphicality conditions of
Theorem~\ref{thm:exact} for the selected profile \(T\). If they fail,
return infeasible.

\item For every diagonal layer \((\alpha,\alpha)\), realize a simple
graph on \(B_\alpha\) with degree sequence $(T_{i\alpha})_{i\in B_\alpha}$ using heap-based Havel--Hakimi.

\item For every off-diagonal layer \((\alpha,\beta)\), with
\(\alpha<\beta\), realize a simple bipartite graph between
\(B_\alpha\) and \(B_\beta\) with left and right degree sequences $(T_{i\beta})_{i\in B_\alpha},
    \qquad
    (T_{j\alpha})_{j\in B_\beta},$ using heap-based bipartite Havel--Hakimi.

\item Unite all diagonal and off-diagonal layer edge sets, and set $\hat A_{\beta\alpha}=\hat A_{\alpha\beta}^{\top}$ for every \(\alpha<\beta\).

\item To generate \(S\) randomized realizations, run layer-preserving
edge switches. First burn in with \(b\) switches per edge, then output
samples separated by \(\theta\) switches per edge. Every accepted switch
preserves the selected profile \(T\).
\end{enumerate}
\end{algorithmblock}

\begin{proposition}[Realization cost]
\label{prop:complexity}
Let $\mathcal L_T$ be the set of diagonal and off-diagonal layers with
nonzero prescribed degrees. For a layer $\ell$, let $n_\ell$ be the number
of vertices incident to the layer and $m_\ell$ the number of edges prescribed
inside it. A fixed-profile realization can be constructed in
\[
    O\!\left(
    \sum_{\ell\in\mathcal L_T}(n_\ell+m_\ell)\log n_\ell
    \right)
\]
time and $O(n+m)$ memory, where $m=\sum_{\ell\in\mathcal L_T}m_\ell$.
Consequently the dense worst-case bound is $O((nk+m)\log n)$, while for
sparse block-interaction patterns with bounded average number of active target
blocks per vertex it is $O((n+m)\log n)$. If $S$ randomized samples are
generated by layer-preserving switches after construction, with burn-in $b$
and thinning $\theta$ accepted switches per edge, the switch chain adds
$O((b+S\theta)m)$ accepted switch operations.
\end{proposition}

\begin{proof}
In a heap-based Havel--Hakimi realization, every produced edge causes a
constant number of heap removals or insertions, each costing
$O(\log n_\ell)$ in its layer. Initial heap construction costs
$O(n_\ell)$. Summing over independent diagonal and off-diagonal layers gives
the stated bound. The memory bound follows because the algorithm stores the
partition, the active heaps and the realized edge set. The dense and sparse
corollaries are obtained by bounding the repeated appearance of vertices
across active block pairs. The switch-chain term is linear in the number of
accepted local edge switches requested per edge.
\end{proof}

\begin{figure}[pos=h]
    \centering
    \includegraphics[width=0.75\linewidth]{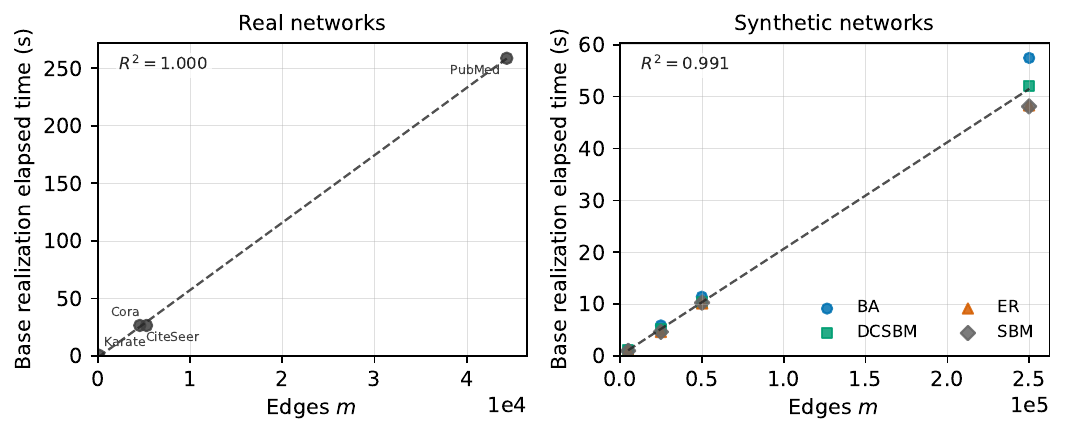}
    \caption{$(AH,H)$ realization time vs. number of edges $m$ for real and synthetic graphs. We report the $R^2$ of the linear fit.}
    \label{fig:scalability}
\end{figure}

Averaging over layers, we get a time complexity $O(n+m)\log n$, which, in case of sparse graphs with $m\propto n$, becomes $\tilde{O}(m)$. In other terms, suppressing polylogarithmic factors, we can state that our proposed algorithm has near-linear scalability in the number of edges $m$. We show this empirically in Figure \ref{fig:scalability}.

\subsection{Integer feasibility profile selector}

Here we provide the details of our choice for the integer feasibility profile selector $\Phi_\eps(H,Q,\varepsilon)$. We highlight that other selectors can
be substituted without changing Theorem~\ref{thm:approx}.

Let $\zeta_{J}(x)$ denote the integer in a finite interval
$J\subset\mathbb Z$ closest to $x$, with ties broken toward the smaller integer.
Let $\delta(n,M,L,U)$ be the nearly constant integer sequence of length $n$, sum $M$, and entries in $[L,U]$: if
$nL\leq M\leq nU$, set $a=\lfloor M/n\rfloor$ and
$r=M-na$, then use $r$ entries equal to $a+1$ and $n-r$ entries equal to $a$, after checking that the resulting sequence lies in $[L,U]$.

\begin{algorithmblock}{Canonical integer-profile selector \(\Phi_\varepsilon(H,Q)\)} \label{alg:profile-selector} \textbf{Input:} partition \(\{B_1,\ldots,B_k\}\) with membership matrix $H$; quotient matrix \(Q\); tolerance \(\varepsilon\).
\textbf{Output:} an integer node-to-block profile \(T\in\mathbb Z_{\ge0}^{n\times k}\), or None if infeasible. 

\begin{enumerate}[leftmargin=18pt]
\item Let \(s_\alpha=|B_\alpha|\) for every block \(B_\alpha\). Initialize \(T\in\mathbb Z_{\ge0}^{n\times k}\). 

\item For every ordered block pair \((\alpha,\beta)\), compute the integer interval $ L_{\alpha\beta} = \max\{0,\lceil Q_{\alpha\beta}-\varepsilon\rceil\}, \qquad U_{\alpha\beta} = \min\{N_{\alpha\beta}, \lfloor Q_{\alpha\beta}+\varepsilon\rfloor\}$, where $N_{\alpha\alpha}=s_\alpha-1, \qquad N_{\alpha\beta}=s_\beta \quad\text{for }\alpha\neq\beta.$
If \(L_{\alpha\beta}>U_{\alpha\beta}\) for some \((\alpha,\beta)\), return None.

\item For every diagonal layer \((\alpha,\alpha)\), compute the admissible internal edge-count interval \[ J_{\alpha\alpha} = \left[ \left\lceil \frac{s_\alpha L_{\alpha\alpha}}{2}\right\rceil, \left\lfloor \frac{s_\alpha U_{\alpha\alpha}}{2}\right\rfloor \right] \cap \left[0,\binom{s_\alpha}{2}\right] \cap\mathbb Z. \]

\item If \(J_{\alpha\alpha}=\varnothing\), return None. Otherwise set $m_{\alpha\alpha} = \zeta_{J_{\alpha\alpha}} \!\left( \frac{s_\alpha Q_{\alpha\alpha}}{2} \right), \quad (T_{i\alpha})_{i\in B_\alpha} = \delta \left( s_\alpha, 2m_{\alpha\alpha}, L_{\alpha\alpha}, U_{\alpha\alpha} \right).$

\item For every off-diagonal layer \((\alpha,\beta)\), with \(\alpha<\beta\), compute the admissible edge-count interval \[ J_{\alpha\beta} = \left[ \max\{s_\alpha L_{\alpha\beta}, s_\beta L_{\beta\alpha}\}, \min\{s_\alpha U_{\alpha\beta}, s_\beta U_{\beta\alpha}, s_\alpha s_\beta\} \right]\cap\mathbb Z. \]

\item If \(J_{\alpha\beta}=\varnothing\), return None. Otherwise set\\
$m_{\alpha\beta} = \zeta_{J_{\alpha\beta}} \!\left( \frac{ s_\alpha Q_{\alpha\beta} + s_\beta Q_{\beta\alpha} }{2} \right), \quad (T_{i\beta})_{i\in B_\alpha} = \delta \left( s_\alpha, m_{\alpha\beta}, L_{\alpha\beta}, U_{\alpha\beta} \right), \quad (T_{j\alpha})_{j\in B_\beta} = \delta \left( s_\beta, m_{\alpha\beta}, L_{\beta\alpha}, U_{\beta\alpha} \right).$

\item Test the block-wise graphicality conditions of Theorem~\ref{thm:exact}. If they fail, return None.

\item Return \(T\).

\end{enumerate}
\end{algorithmblock}

\section{Edge overlap as a reconstruction privacy diagnostic}
\label{sec:edge-overlap}

Let
$
    \EO(A,A')=|E(A)\cap E(A')|/|E(A)|,
$
where the vertex correspondence between $A$ and $A'$ is induced by a
specified matching. Here we choose to match vertices by ranking in the degree sequence. Low edge overlap means that few original
edges are recovered under the stated attack; this is desirable in our setting,
provided the reconstructed graph still reproduces the graph properties of
interest. This diagnostic measures edge-reconstruction risk and is not a
differential-privacy guarantee.
The overlap prediction is expressed in terms of the integer profile $T$ that
is actually realized.
Let
$
    \mathcal L
    =
    \{(\alpha,\alpha):1\leq \alpha\leq k\}
    \cup
    \{(\alpha,\beta):1\leq \alpha<\beta\leq k\}
$
be the set of diagonal and off-diagonal layers.
\begin{wrapfigure}[20]{r}{0.5\textwidth}
    \vspace{-0.25cm}
    \centering
    \includegraphics[width=\linewidth]{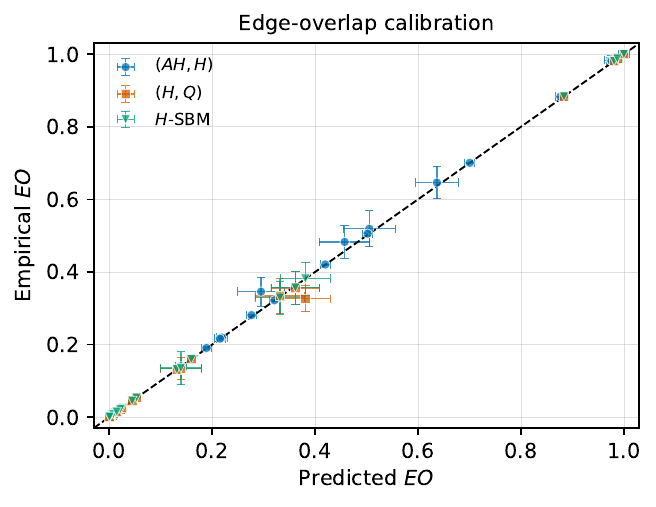}
    \caption{Empirical vs. predicted edge overlap for reconstruction methods $H$-SBM and $\varepsilon$-EP with releases $(AH,H)$ and $(H,Q)$.}
    \label{fig:eo_calibration}
\end{wrapfigure}
For each layer
$\ell\in\mathcal L$, let $N_\ell$ be the number of possible edges in that
layer:
\[
    N_{\alpha\alpha}=\binom{s_\alpha}{2},
    \qquad
    N_{\alpha\beta}=s_\alpha s_\beta
    \quad(\alpha<\beta).
\]
Let $m_\ell^A$ be the number of original edges of $A$ in layer $\ell$,
and let $m_\ell^T$ be the number of edges prescribed by the selected profile
$T$. Explicitly,
\[
    m_{\alpha\alpha}^A
    =
    \frac{1}{2}\sum_{i\in B_\alpha}(AH)_{i\alpha},
    \qquad
    m_{\alpha\alpha}^T
    =
    \frac{1}{2}\sum_{i\in B_\alpha}T_{i\alpha},
\]
and, for $\alpha<\beta$,
\[
    m_{\alpha\beta}^A
    =
    \sum_{i\in B_\alpha}(AH)_{i\beta},
    \qquad
    m_{\alpha\beta}^T
    =
    \sum_{i\in B_\alpha}T_{i\beta}
    =
    \sum_{j\in B_\beta}T_{j\alpha}.
\]
The last equality follows from feasibility of $T$.

Define the layer density
induced by $T$ as
$
    p_\ell^T=\frac{m_\ell^T}{N_\ell},
$
with $p_\ell^T=0$ when $N_\ell=0$. The expected profile-conditioned edge-overlap
is
\begin{equation}
    \widehat{\EO}(A,T;H)
    =
    \frac{1}{|E(A)|}
    \sum_{\ell\in\mathcal L}
    m_\ell^A p_\ell^T
    =
    \frac{1}{|E(A)|}
    \sum_{\ell\in\mathcal L}
    m_\ell^A \frac{m_\ell^T}{N_\ell}.
\label{eq:EO_profile}
\end{equation}
Equation~\eqref{eq:EO_profile} is the expectation obtained when, conditional
on the selected layer edge counts, reconstructed edges are treated as uniformly
placed inside each layer.
We also report a fixed-layer hypergeometric
uncertainty proxy. For a general selected profile $T$, we use
\begin{equation}
\widehat{\SE}_{\mathrm{hyp}}(\EO)
=
\frac{1}{|E(A)|}
\sqrt{
\sum_{\ell\in\mathcal L}
\frac{
m_\ell^A m_\ell^T
(N_\ell-m_\ell^A)
(N_\ell-m_\ell^T)
}{
N_\ell^2(N_\ell-1)
}
},
\label{eq:EO_SE_general}
\end{equation}
with zero contribution when $N_\ell\leq1$.

\section{Numerical experiments}
\label{sec:experiments}

In this section, we outline the numerical experiments performed to evaluate our $\varepsilon$-EP realization algorithm in terms of both utility and privacy preservation.

\subsection{Datasets}
We evaluate our $\varepsilon$-EP realization algorithm on four undirected real networks: Zachary's Karate
Club network \cite{Zachary1977} and the Cora, CiteSeer and PubMed citation
networks \cite{Yang2016Planetoid}.

For the scalability experiment we generate synthetic graphs in four families:
Erd\H{o}s--R\'enyi graphs \cite{erdos59a}, Barab\'asi--Albert graphs \cite{doi:10.1126/science.286.5439.509}, four-block stochastic
block models \cite{HOLLAND1983109}, and four-block degree-corrected stochastic block models \cite{PhysRevE.83.016107}. The size grid was $n\in\{10^3,5\cdot 10^3,10^4,5\cdot 10^4\}$, with target average degree $10$, so all cases are sparse and have $m\approx 5n$. For the block models, the mixing parameter is $0.15$; in the degree-corrected case, within-block propensities are Pareto distributed.

\subsection{$\varepsilon$-EP algorithm and other baselines}
For each graph we find the $\varepsilon$-EPs with the $\varepsilon$-BE algorithm of Squillace et
al. \cite{squillace2024efficient,squillace2026scalable}. We set $\varepsilon$ from percentiles $p\in\{0,25,50,75,100\}$ of the degree sequence to track the effects of coarsening the partitions.

For each value of $\varepsilon$ we test the algorithm with both $(AH,H)$ and $(H,Q)$ releases. We compare them
with three baselines:
\begin{itemize}
    \item Havel--Hakimi \cite{Havel1955,Hakimi1962};
    \item Newman's configuration model \cite{newman_configuration, NewmanStrogatzWatts2001};
    \item a canonical stochastic block model \cite{HOLLAND1983109} using $H$ as indicator matrix. We will name this $H$-SBM, providing further details in Section C of the Supplementary Material.
\end{itemize}
 
While Havel--Hakimi is a deterministic algorithm, i.e. the degree sequence uniquely determines the constructed graph, with $\varepsilon$-EP algorithm and Newman's configuration model we obtain 50 samples per setting, after a burn-in of 50
switches per edge and thinning of 5 switches per edge.

\subsection{Privacy evaluation}
To evaluate the privacy risk of each reconstruction method, we report mean and standard deviation of $1-\EO$, the complement of edge overlap as defined in Section \ref{sec:edge-overlap}. We also show the comparison with the predictors described by Equations \ref{eq:EO_profile} and \ref{eq:EO_SE_general} (see Figure \ref{fig:eo_calibration}).

\subsection{Utility evaluation} \label{sec:utility_eval}
To evaluate the utility of reconstructed graphs, we measure:
\begin{itemize}
    \item average counts of connected components, triangles, wedges, size of largest connected component, global clustering coefficient, assortativity, power-law exponent (see Tables S3--S6 of the Supplementary Material). Since these are scalar quantities, for the utility we compute the relative error;
    \item mean Jensen--Shannon divergence of eigenvector, betweenness, PageRank and Katz centralities (see Figure \ref{fig:main-error-profiles}). For $\varepsilon$-EP reconstructions, we also report the reconstructed vs. original centrality distributions in Figures S8--S15 of the Supplementary Material.
\end{itemize}

\begin{figure}[pos=h]
    \centering
    \includegraphics[width=0.6\linewidth]{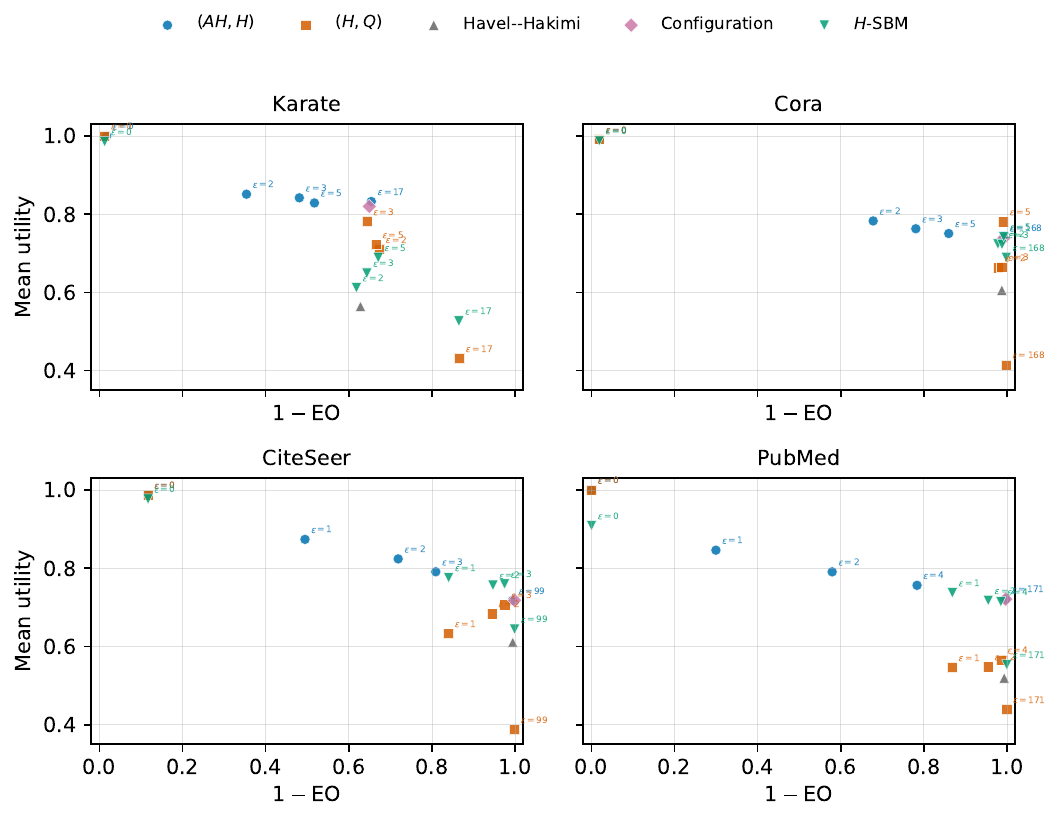}
    \caption{On the x-axis, the complement of the edge overlap, as defined in Section \ref{sec:edge-overlap}. On the y-axis, the mean utility, as defined in Section \ref{sec:utility_eval}. $\varepsilon$-EP with $(AH,H)$ release and $H$-SBM show the best privacy-utility trade-off, with $\varepsilon$-EP capable of reconstructing networks with higher utility.}
    \label{fig:utility_eo}
\end{figure}

\begin{figure}[pos=h]
\centering
\includegraphics[width=0.64\textwidth]{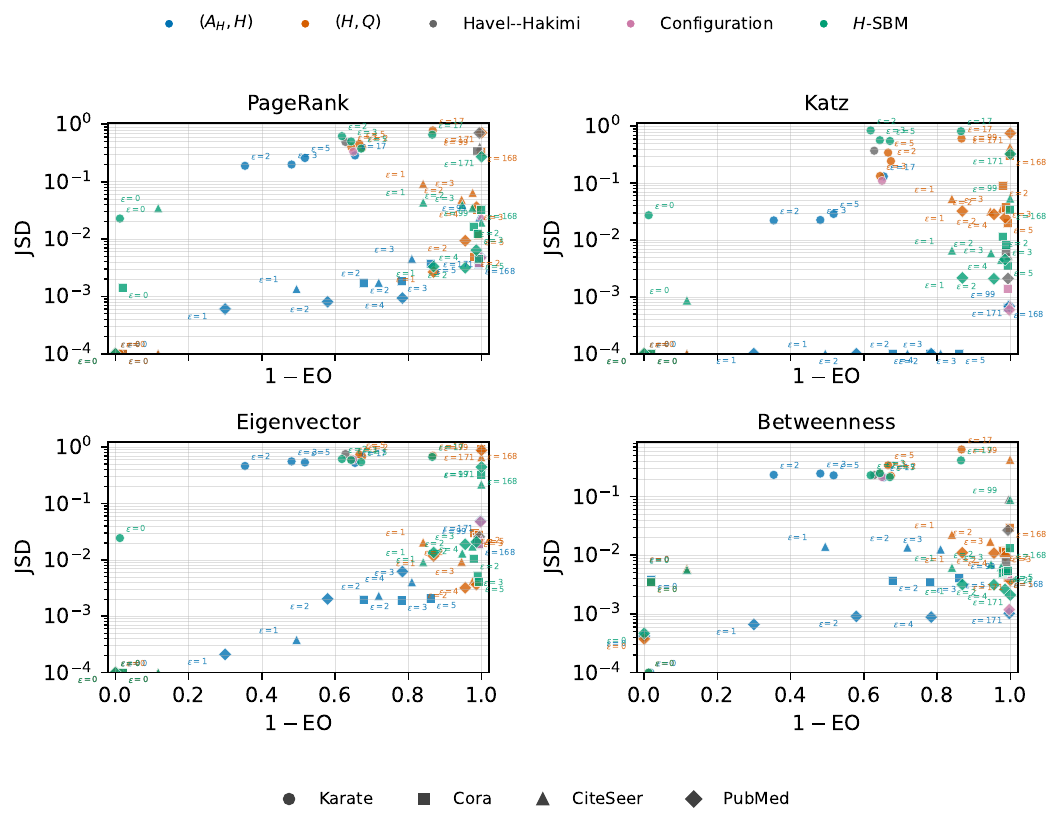}
\caption{Jensen--Shannon divergence against the privacy proxy $1-\mathrm{EO}$. Each point is a reconstructed graph, coloured by method and marked by dataset. For $(A_H,H)$, $(H,Q)$ and $H$-SBM, annotations report the $\varepsilon$ value. JSD values below $10^{-4}$ are displayed on the $10^{-4}$ line for visibility.}
\label{fig:main-error-profiles}
\end{figure}

We report the complete results in Section D of the Supplementary Material, while as a summary in the main text we report in Figure \ref{fig:utility_eo} a single mean utility score, defined as $1$
minus the average error over the above metrics.

Furthermore, triangle counts are not preserved by EP constraints. This is a well-known limitation: for example, 1-WL cannot distinguish a 6-cycle from two
disjoint triangles when all vertices have the same initial colour. Therefore, in Figure S1 of the Supplementary Material, we report the triangle relative
error separately.

\subsection{Diffusion error}
To further assess the error on graph dynamic properties by each reconstruction method, in addition to the error on centrality measures, we define the relative diffusion error at time $t$ as
\begin{equation}
    \frac{\|A^tH-\hat A^tH\|_F}{\|A^tH\|_F}.
\end{equation}
As a reference, we compare the diffusion error with the \emph{lifted quotient dynamics}
\begin{equation}
    \frac{\|A^tH- H Q^t\|_F}{\|A^tH\|_F},
\end{equation}
which measures the dynamical information retained by the quotient
itself.
We measure the error on short-term ($t\in\{1,2,3,4,5\}$) and long-term ($t\in\{2,4,8,16,32\}$) diffusion.
Naturally, for $\varepsilon=0$ the relative diffusion error is zero by definition, so we show the results for $\varepsilon$ set to the degree percentiles $25,50,75,100$.
We report in the main text (Figures \ref{fig:short-diffusion-karate} and \ref{fig:long-diffusion-karate}) the results for the Karate dataset, while the others are reported in Figures S2--S7 of the Supplementary Material.

\subsection{Scalability} \label{sec:scalability}
To empirically assess the time complexity of the $\varepsilon$-EP realization algorithm already discussed in Section \ref{sec:algorithm}, we show with both real and synthetic graphs how time scales approximately linearly with the number of edges $m$ \ref{fig:scalability}.
In particular, we report the base constructive $(AH,H)$ realization step, without considering any edge switch (we already showed in Section \ref{sec:algorithm} that this step adds another $O(m)$ term).

\begin{figure}[pos=t]
    \centering
    \includegraphics[width=\linewidth]{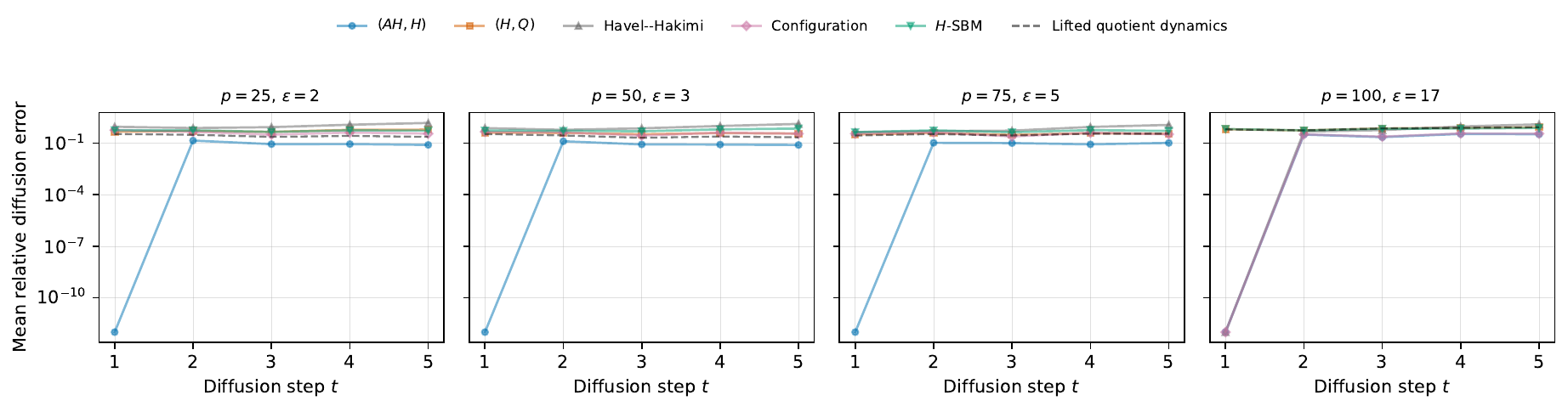}
    \caption{Short-term relative diffusion error for reconstructed Karate dataset. At $t=1$, $\varepsilon$-EP reconstruction with release $(AH,H)$ achieves a mean relative diffusion error lower than other methods by several orders of magnitude. For $t>1$, $\varepsilon$-EP reconstruction with release $(AH,H)$ is still the best method, distancing itself by nearly one order of magnitude from the other methods.}
    \label{fig:short-diffusion-karate}
\end{figure}

\begin{figure}[pos=t]
    \centering
    \includegraphics[width=\linewidth]{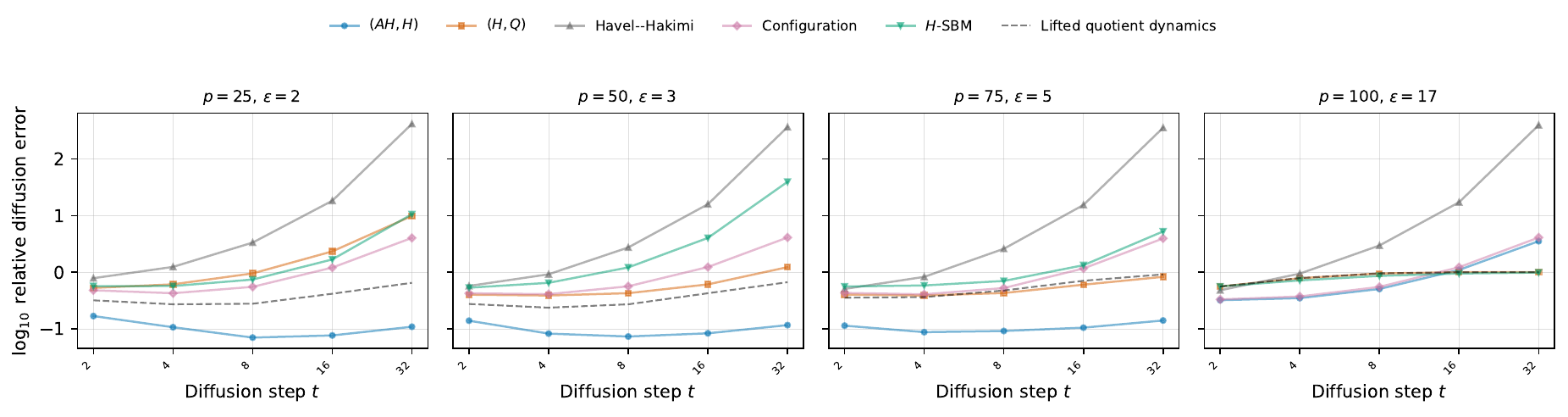}
    \caption{Logarithm of mean long-term relative diffusion error for the reconstructed Karate dataset. $\varepsilon$-EP reconstruction with release $(AH,H)$ is the only method that achieves an error lower than the lifted quotient dynamics, keeping the relative diffusion error at nearly $0.1$. All other methods diverge at $t\geq4$.}
    \label{fig:long-diffusion-karate}
\end{figure}

\section{Discussion}
\label{sec:discussion}

The numerical results support the advantages discussed in this paper for the adoption of approximate EPs as an
intermediate representation between the complete topology of a graph and
the substantially weaker information provided by its degree sequence.

Figure~\ref{fig:utility_eo} shows that increasing $\varepsilon$ generally moves the EP-based reconstructions
towards lower edge overlap, at the price of progressively lower
structural fidelity, naturally, but still preserving better utility then all other methods. In particular, the $(AH,H)$ release provides the most favourable tuning of the privacy--utility trade-off. The $H$-SBM immediately achieves a near-zero edge overlap for $\varepsilon>0$, but with a mean utility always under $0.80$,
whereas the quotient-only $(H,Q)$ release loses substantially more
information as the partition becomes coarse.

The scalar statistics in Tables S3--S6 in the Supplementary Material clarify which part of this utility
comes directly from the information retained by the release. In the
$(AH,H)$ case, the degree of every vertex is preserved because it is the
sum of its block degrees. Consequently, quantities determined by the
degree sequence are preserved as well. This can be seen, for example, in
the exact preservation of the wedge count and of the fitted power-law
exponent across the $(AH,H)$ reconstructions. The method also maintains
assortativity remarkably well over a wide range of $\varepsilon$ values
on the citation networks. At the same time, properties that depend on the precise arrangement of edges (connected components, size of the largest
connected component and clustering coefficient, which are not invariant quantities under EP) progressively deteriorate as the partition is coarsened.

The comparison with degree-sequence methods is particularly informative in this respect: Havel--Hakimi can strongly distort clustering, connectivity and triangle counts, while randomized degree-preserving reconstruction tends towards graphs with very low clustering. Our $\varepsilon$-EP realization algorithm, on the other hand, appears to retain more additional mesoscopic information.

The same distinction becomes even clearer when considering centrality measures. Figure~\ref{fig:main-error-profiles} shows that $\varepsilon$-EP realized graphs, and in particular
the $(AH,H)$ release, achieve the best fidelity to centrality measure even at low edge overlap. It is possible to appreciate this result even better by looking at Figures S8--S15 of the Supplementary Material. Betweenness centrality is not
algebraically determined by the EP in the same way, but its distribution is nevertheless reproduced more accurately than by the degree-only baselines.

The diffusion experiments provide the most direct validation of the original motivation for $\varepsilon$-EP realization problem. On all datasets, $\varepsilon$-EP reconstructions with release $(AH,H)$ at $t=1$ achieve relative diffusion errors
several orders of magnitude below the competing methods. The advantage remains visible for subsequent time steps. More importantly, the long-term
experiment shows that it is the only reconstruction
method whose error remains below that of the lifted quotient dynamics, with an error of approximately $10^{-1}$, whereas the other reconstruction methods rapidly diverge as the number of diffusion steps increases. 

Finally, the scalability results in Figure~\ref{fig:scalability} support the complexity
analysis of Section~\ref{sec:scalability}. The realization time grows approximately linearly with the number of edges $m$ over both the real-world networks and the
four families of synthetic graphs: the least-squares fit of elapsed time against $m$ gives $R^2=1.000$ and $R^2=0.991$,
respectively.

\section{Conclusion}
\label{sec:conclusion}

In this work, we introduced EP-realizability, a graph realization
problem in which the available information is not a degree
sequence but an (approximate) equitable partition.
Not only we prove that the defined problem has a solution, but the proof also constructively provides an algorithm whose complexity is near-linear in the number of edges of the graph to be realized.

This formulation provides a direct connection between graph realization and the dynamical information encoded by equitable partitions, such as eigenvector, PageRank and Katz centralities.

We also derived an analytical predictor for edge overlap between the original and the EP-reconstructed graphs, showing
empirically its correctness.
Experiments show that approximate
EP-realizability provides a tunable reconstruction trade-off:
coarser partitions progressively reduce the recovery of original edges (preserving privacy), 
while retaining substantially more structural and
dynamical utility than degree-only reconstruction or block models (preserving utility). This advantage is particularly visible for centrality
distributions and linear diffusion dynamics.

Several directions naturally follow from these results. First, we aim at ways of improving our algorithm so that it can preserves quantities that are not immediately invariant under EP, like the triangle count and therefore other dependent quantities like global clustering coefficient. For example, the $\varepsilon$-BE algorithm by Squillace et al. \cite{squillace2024efficient, squillace2026scalable} allows to set an initial pre-partition, so it is possible to isolate components or introduce any type of bias before the coursening.

Second, the present work considers simple undirected unweighted graphs,
but extending EP-realizability to directed and weighted networks is a
natural next step. The $\varepsilon$-BE algorithm already
provides $\varepsilon$-EPs in both of these settings, so
the partition-construction stage does not constitute an obstacle.

More generally, the results of this work suggest that an equitable partition can be regarded not only as a tool for graph reduction or dynamical analysis, as it as already been done in the literature \cite{o2013observability,bonaccorsi2015epidemic, scholkemper2023optimization, squillace2024efficient, squillace2026scalable}, but also as a generative model: EP-realizability serves as the
combinatorial substrate for new types of configuration models, opening the way to dynamics-preserving network reconstruction and synthetic graph generation methods.

\printcredits

\section*{Declaration of generative AI and AI-assisted technologies in the writing process}
During the preparation of this work the author used AI-assisted tools to improve the writing of the manuscript. After using these tools, the author reviewed and edited the content as needed and takes full responsibility for the content of the publication.

\bibliographystyle{cas-model2-names}
\bibliography{ep_realizability_final}

\bio{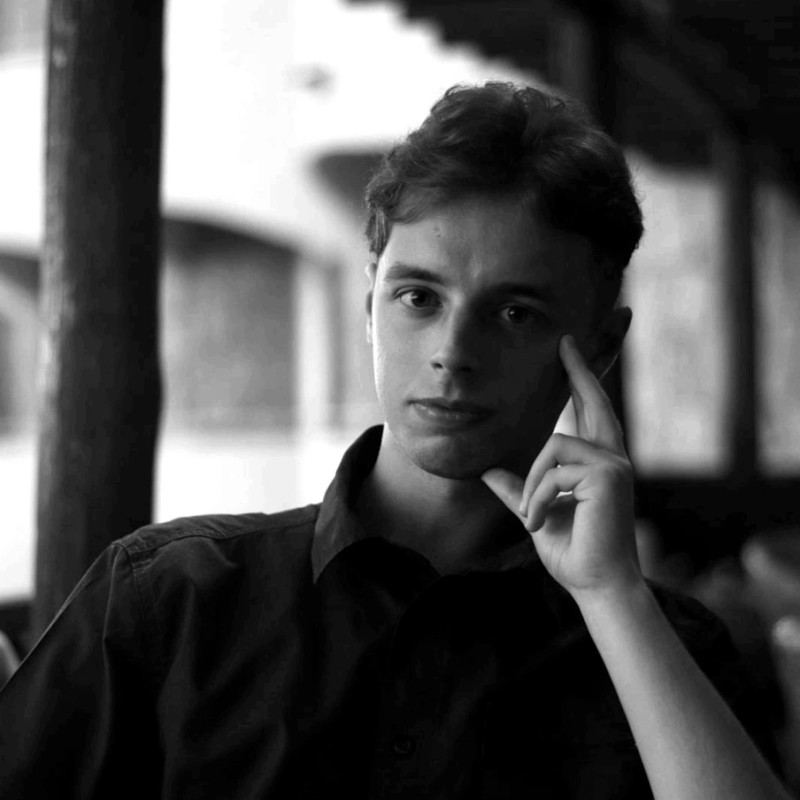}
Riccardo Porcedda is a PhD candidate in the Italian National Ph.D. Program in Artificial Intelligence for Society, with a scolarship from the Scuola Superiore Sant’Anna, Pisa, Italy. His research focuses mostly on graph neural networks, synthetic network generation, equitable partitions and evaluation of LLMs. He holds a bachelor's degree in Physics at the University of Pisa and a master's degree in Data Science at the University of Milano-Bicocca. He has been a visiting researcher at the Complexity Science Hub in Vienna and at the Computational Netowrk Science group of the RWTH Aachen University. He also had experiences in startups, working in the field of xAI and liquidity optimization in financial networks (entering also the Fintech Milano Hub of Banca d'Italia for the project "B2Bridge" while being Chief Data Officer of Liqex Srl). In December 2025 he founded Little-g.AI, a startup focused on the development of efficient foundational models alternatives to LLMs for applications on the academic publishing processes.
\endbio

\end{document}


\title{Supplementary Material for\\Equitable Partition Realizability for Dynamics-preserving and Privacy-aware Network Reconstruction}

\author{Riccardo Porcedda}

\date{}

\maketitle

\appendix

\section{Relation between $\varepsilon$-BE and $\varepsilon$-EP}
\label{app:epsbe-ep}

The approximate equitable partitions used in this work are computed
using the $\varepsilon$-BE partition-refinement algorithm introduced
by Squillace et al.~\cite{squillace2024efficient, squillace2026scalable}.
In this section, we clarify their
relationship and show that $\varepsilon$-BE provides a sufficient,
but not necessary, condition for our definition of $\varepsilon$-EP.

Let $A=A^\top\in\{0,1\}^{n\times n}$ be the adjacency matrix of a
simple undirected graph and let $\mathcal{P}=\{B_1,\ldots,B_k\}$
be a partition of its nodes, represented by the membership matrix
$H\in\{0,1\}^{n\times k}$. Recall that $(AH)_{i\beta} = \sum_{j\in B_\beta}A_{ij}$ is the number of neighbours of node $i$ belonging to block $B_\beta$.
For convenience, we write $d_{i\beta}:=(AH)_{i\beta}$.

For undirected graphs, Squillace et al.~\cite{squillace2024efficient, squillace2026scalable}
define a partition $\mathcal{P}$ to be an $\varepsilon$-BE partition
if, for every pair of blocks $B_\alpha,B_\beta$ and every pair of
nodes $i,j\in B_\alpha$,
\begin{equation}
    \left|
    d_{i\beta}-d_{j\beta}
    \right|
    \leq \varepsilon.
    \label{eq:eps-be-def}
\end{equation}
Thus, within each source block $B_\alpha$, the block-degrees towards
every target block $B_\beta$ are required to have range at most
$\varepsilon$.
We already defined the quotient $Q=(H^\top H)^{-1}H^\top AH$ and, since $H^\top H=\operatorname{diag}(|B_1|,\ldots,|B_k|)$, its entries
are
\begin{equation} \label{eq:qavg-blockmean}
Q_{\alpha\beta}
    =
    \frac{1}{|B_\alpha|}
    \sum_{j\in B_\alpha}d_{j\beta}.
\end{equation}
Hence, $Q_{\alpha\beta}$ is the average number of
neighbours in $B_\beta$ of the nodes belonging to $B_\alpha$.
We call $H$ an $\varepsilon$-EP whenever
\begin{equation}
    \left\|
    AH-HQ
    \right\|_\infty
    \leq\varepsilon,
    \label{eq:eps-ep-def-app}
\end{equation}
or, equivalently,
\begin{equation}
    \left|
    d_{i\beta}-Q_{\alpha\beta}
    \right|
    \leq\varepsilon
    \qquad
    \forall i\in B_\alpha,\ \forall\alpha,\beta.
\end{equation}
Unlike $\varepsilon$-BE, which controls the pairwise diameter of the
block-degree values, the inequality~\eqref{eq:eps-ep-def-app} controls
their distance from the block mean.

\begin{proposition}
\label{prop:epsbe-implies-ep}
If $H$ is the membership matrix of an $\varepsilon$-BE partition of $A$, then it is also the membership matrix of an
$\varepsilon$-EP according to inequality~\eqref{eq:eps-ep-def-app}.
That is,
\begin{equation*}
    H\text{ is }\varepsilon\text{-BE}
    \quad\Longrightarrow\quad
    \left\|AH-HQ\right\|_\infty
    \leq\varepsilon.
\end{equation*}
The converse does not hold in general.
\end{proposition}

\begin{proof}
Fix two blocks $B_\alpha$ and $B_\beta$ and a node
$i\in B_\alpha$. From~\eqref{eq:qavg-blockmean},
\begin{align*}
\left|
d_{i\beta}-Q_{\alpha\beta}
\right|
&=
\left|
d_{i\beta}
-
\frac{1}{|B_\alpha|}
\sum_{j\in B_\alpha}d_{j\beta}
\right|
\\
&=
\left|
\frac{1}{|B_\alpha|}
\sum_{j\in B_\alpha}
(d_{i\beta}-d_{j\beta})
\right|
\\
&\leq
\frac{1}{|B_\alpha|}
\sum_{j\in B_\alpha}
\left|d_{i\beta}-d_{j\beta}\right|.
\end{align*}
If $H$ is the membership matrix of an $\varepsilon$-BE partition, each term in the final sum
is at most $\varepsilon$ by~\eqref{eq:eps-be-def}. Therefore,
\[
\left|
d_{i\beta}-Q_{\alpha\beta}
\right|
\leq
\frac{1}{|B_\alpha|}
\sum_{j\in B_\alpha}\varepsilon
=
\varepsilon.
\]
Since this holds for every $i,\alpha,\beta$, then $\left\|AH-HQ\right\|_\infty
\leq\varepsilon$.

To show that the converse is false, consider a block containing four
nodes whose block-degrees towards some block $B_\beta$ are 
$
(d_{1\beta},d_{2\beta},d_{3\beta},d_{4\beta})
=
(0,0,2,2).
$
Their block average is $Q_{\alpha\beta}=1$.
Consequently, $\max_{i\in B_\alpha}
\left|
d_{i\beta}-Q_{\alpha\beta}
\right|
=1,
$ and hence this coordinate satisfies the $1$-EP condition. However, $\max_{i,j\in B_\alpha}
|d_{i\beta}-d_{j\beta}|
=
2,
$ which violates the $1$-BE condition. Therefore, a $1$-EP need not be a $1$-BE.
\end{proof}





This observation justifies the use of the $\varepsilon$-BE algorithm
of Squillace et al.~\cite{squillace2024efficient, squillace2026scalable} to compute
the membership matrix $H$ in our framework. Any partition returned
by the algorithm with tolerance $\varepsilon$ automatically satisfies
our $\varepsilon$-EP condition with the same tolerance. Therefore,
the algorithm provides a sufficient constructive procedure for
obtaining admissible approximate EPs.

\section{A Toy Example of Exact and Approximate EP Realizability}
\label{app:ep-realizability-example}

This section explains in detail the example depicted in Figure 1, where we show how an exact EP can uniquely determine a graph, while relaxing the equitability condition may lead to multiple non-isomorphic realizations.

\subsection{Exact EP and unique realization}

Let us consider a simple, undirected, unweighted graph $G_0=(V,E)$ with $V=\{1,\ldots,8\}$, and the partition $\mathcal{P}_0=\{C_1,C_2,C_3,C_4\}$, where $C_1=\{1,2\},\quad
C_2=\{3,4\},\quad
C_3=\{5,6\},\quad
C_4=\{7,8\}.
$

The corresponding characteristic matrix is

$$
H_0=
\begin{pmatrix}
1&0&0&0\\
1&0&0&0\\
0&1&0&0\\
0&1&0&0\\
0&0&1&0\\
0&0&1&0\\
0&0&0&1\\
0&0&0&1
\end{pmatrix}.
$$

Let the quotient matrix be

$$
Q_0=
\begin{pmatrix}
0&0&0&0\\
0&0&0&2\\
0&0&1&2\\
0&2&2&1
\end{pmatrix}.
$$

Therefore, the exact EP condition is

$$
A_0H_0=H_0Q_0.
$$

The pair $(H_0,Q_0)$ uniquely determines the graph, up to relabeling. To see this, notice that all cells have size two. Hence, for every off-diagonal block $A_{ij}$, $i\neq j$, an entry $q_{ij}=0$ forces the corresponding bipartite block to be empty, whereas $q_{ij}=2=|C_j|$ forces it to be complete bipartite. Similarly, for a diagonal block of size two, $q_{ii}=0$ forces the two vertices to be disconnected, while $q_{ii}=1$ forces the unique possible internal edge.

Consequently,
\begin{itemize}
\item $C_1$ has no internal or external edges;
\item $C_2$ has no internal edge and is completely connected to $C_4$;
\item $C_3$ contains its unique internal edge and is completely connected to $C_4$;
\item $C_4$ contains its unique internal edge.
\end{itemize}

Thus, the unique adjacency matrix is

$$
A_0=
\begin{pmatrix}
0&0&0&0&0&0&0&0\\
0&0&0&0&0&0&0&0\\
0&0&0&0&0&0&1&1\\
0&0&0&0&0&0&1&1\\
0&0&0&0&0&1&1&1\\
0&0&0&0&1&0&1&1\\
0&0&1&1&1&1&0&1\\
0&0&1&1&1&1&1&0
\end{pmatrix}.
$$




\subsection{Relaxing the EP}

We now relax equitability condition and consider an approximate equitable partition with $\varepsilon=1$. The four exact cells can be merged into two larger cells, $B_1=C_1\cup C_2=\{1,2,3,4\},
\quad
B_2=C_3\cup C_4=\{5,6,7,8\}$.

The corresponding characteristic matrix is

$$
H_1=
\begin{pmatrix}
1&0\\
1&0\\
1&0\\
1&0\\
0&1\\
0&1\\
0&1\\
0&1
\end{pmatrix},
$$
for which we get the quotient $Q_1=(H_1^\top H_1)^{-1}H_1^\top A_0H_1=
\begin{pmatrix}
0&1\\
1&3
\end{pmatrix}.
$
Therefore, it is easy to verify that $\left\|A_0H_1-H_1Q_1\right\|_{\infty}=1$ and $\mathcal{P}_1=\{B_1,B_2\}$ is a valid $1$-EP.














\subsection{Characterization of all realizations of the same $1$-EP}

We now  ask how many simple undirected graphs $A$ satisfy $\|AH_1-H_1Q_1\|_\infty\leq 1$.
The quotient entries $q_{\alpha \beta}$ impose strong constraints on the two diagonal blocks. Since $q_{11}=0$, the average internal degree in $B_1$ is zero. As all degrees are non-negative, this implies $G[B_1]=\overline{K}_4$.
Similarly, $q_{22}=3$. Since a vertex in a four-node simple graph can have at most three internal neighbors, every vertex in $B_2$ must have internal degree three, and therefore $G[B_2]=K_4$.
Finally, $q_{12}=q_{21}=1$ implies that the number of edges between $B_1$ and $B_2$ is exactly four.

Hence, every realization has the block form
$$
A=
\begin{pmatrix}
0_{4\times4} & X\\
X^\top & J_4-I_4
\end{pmatrix},
$$
where $X\in\{0,1\}^{4\times4}$ is the incidence matrix of the bipartite graph between $B_1$ and $B_2$.
Because the target cross-cell degree is $1$ and $\varepsilon=1$, every row and every column of $X$ must have sum in $\{0,1,2\}$.
Moreover, $\sum_{i,j}X_{ij}=4$, thus, the complete realization space is
$$
\mathcal{R}
=
\left\{
X\in\{0,1\}^{4\times4}:
\sum_{i,j}X_{ij}=4,\;
\max_i\sum_jX_{ij}\leq2,\;
\max_j\sum_iX_{ij}\leq2
\right\}.
$$

\subsection{Non-isomorphic realizations}

The bipartite graph associated with $X$ contains four edges and has maximum degree two. Therefore, each of its non-isolated connected components must be either a path or an even cycle.

There is an additional useful observation. In the full graph, every vertex in $B_1$ has degree at most two, whereas every vertex in $B_2$ has degree at least three, because $B_2$ induces a $K_4$. Consequently, any graph isomorphism necessarily maps $B_1$ onto $B_1$ and $B_2$ onto $B_2$. The relevant equivalence relation on $X$ is therefore given by independent row and column permutations, i.e., $X\sim P_{B_1}X P_{B_2}^\top,
\quad
P_{B_1},P_{B_2}\in S_4.
$
Under this equivalence relation, there are exactly ten non-isomorphic realizations, reported in Table S\ref{tab:approx-ep-realizations}. The original graph $A_0$ belongs to the first class.

\begin{table}[ht]
\centering
\small
\caption{The ten non-isomorphic realizations of the cross-cell bipartite graph. The final column gives the number of distinct labeled realizations in the corresponding orbit under independent permutations of $B_1$ and $B_2$.}
\label{tab:approx-ep-realizations}
\begin{tabular}{c|l|r}
\hline
Class & Structure & labeled graphs\\
\hline
1 &
$C_4$ &
36

\\

2 &
$P_5$, $2B_1+3B_2$ &
144
\\

3 &
$P_5$, $3B_1+2B_2$ &
144
\\

4 &
$P_4\cup P_2$ &
576
\\

5 &
$P_3\cup P_3$, centers in $B_1$ &
36
\\

6 &
$P_3\cup P_3$, centers split &
144
\\

7 &
$P_3\cup P_3$, centers in $B_2$ &
36
\\

8 &
$P_3\cup2P_2$, center in $B_1$ &
144
\\

9 &
$P_3\cup2P_2$, center in $B_2$ &
144
\\

10 &
$4P_2$ &
24
\\
\hline
\end{tabular}
\end{table}







\subsection{Number of labeled and non-isomorphic realizations}

The ten rows of Table S\ref{tab:approx-ep-realizations} form the ten orbits of admissible bipartite matrices under the action of $S_4\times S_4$.
Thus, when vertex labels are ignored, the approximate EP has 10 non-isomorphic realizations. When all distinct labellings are considered, while identifying permutations that leave the adjacency matrix unchanged through automorphisms, the orbit sizes sum to 1428.
Equivalently,
$$
\left|
\left\{
X\in\{0,1\}^{4\times4}:
\|X\|_0=4,\;
\max_i(X\mathbf{1})_i\leq2,\;
\max_j(X^\top\mathbf{1})_j\leq2
\right\}
\right|
=
1428.
$$

This example highlights the loss of identifiability induced by approximate equitability. The exact representation $(H_0,Q_0)$ uniquely determines the complete adjacency matrix, while the relaxed representation $(H_1,Q_1)$ preserves a meaningful two-dimensional quotient, but admits a large set of topologically distinct realizations. Hence, increasing the EP tolerance can substantially enlarge the realization space without necessarily collapsing the quotient.

\section{Details about \(H\)-SBM}
\label{app:h-sbm-baseline}

For completeness, we describe the canonical partition-conditioned stochastic
block model. Let $H$ be the membership matrix of the partition
\(\{B_1,\ldots,B_k\}\), with \(s_\alpha=|B_\alpha|\). For every unordered block
pair \((\alpha,\beta)\), let \(m_{\alpha\beta}\) be the number of edges of the
original graph with one endpoint in \(B_\alpha\) and one endpoint in
\(B_\beta\). When \(\alpha=\beta\), \(m_{\alpha\alpha}\) denotes the number of
internal edges in \(B_\alpha\).

We define the probability matrix
\(P\) as
\[
P_{\alpha\beta}
=
\begin{cases}
\dfrac{m_{\alpha\alpha}}{\binom{s_\alpha}{2}},
& \alpha=\beta,\ s_\alpha\geq 2,\\[1.2em]
\dfrac{m_{\alpha\beta}}{s_\alpha s_\beta},
& \alpha\neq\beta,\\[1.2em]
0,
& \text{if the corresponding layer has zero capacity.}
\end{cases}
\]
Since we are dealing with undirected graphs, \(P_{\alpha\beta}=P_{\beta\alpha}\).
A graph \(\hat A\) is then sampled independently over all
admissible node pairs, with $\Pr[\hat A_{ij}=1]=P_{\alpha\beta},
\quad i\in B_\alpha,\ j\in B_\beta,\ i\neq j$.
Thus, within each block \(B_\alpha\), the model samples an Erd\H{o}s--R\'enyi
graph with parameter \(P_{\alpha\alpha}\), while between two distinct blocks it
samples an independent bipartite Bernoulli graph with parameter
\(P_{\alpha\beta}\).
This baseline preserves the expected number of edges in every partition layer:
\[
\mathbb E[\hat m_{\alpha\alpha}]
=
\binom{s_\alpha}{2}P_{\alpha\alpha}
=
m_{\alpha\alpha},
\qquad
\mathbb E[\hat m_{\alpha\beta}]
=
s_\alpha s_\beta P_{\alpha\beta}
=
m_{\alpha\beta}
\quad (\alpha\neq\beta).
\]

\clearpage

\section{Additional experiment results}
\label{app:experiments}

\begin{table}[h]
\centering
\caption{$\varepsilon$-BE partition diagnostics. The column $p$ is the requested percentile of the degree sequence used to set $\varepsilon$; $\|AH-HQ\|_\infty$ is the maximum node-to-block residual.}
\label{tab:partition-summary}
\scriptsize
\resizebox{\textwidth}{!}{%
\begin{tabular}{llrrrrrrrl}
\toprule
Dataset & $p$ & $\varepsilon$ & $k$ & $k/n$ & Min & Med. & Max & Singleton nodes (\%) & $\|AH-HQ\|_\infty$ \\
 & & & & & \multicolumn{3}{c}{Block size} & & \\
\midrule
\multirow{5}{*}{Karate} & 0 & 0 & 27 & 0.7941 & 1 & 1 & 5 & 67.6 & 0 \\
 & 25 & 2 & 6 & 0.1765 & 1 & 2 & 18 & 5.9 & 1.556 \\
 & 50 & 3 & 4 & 0.1176 & 2 & 4 & 24 & 0.0 & 2.25 \\
 & 75 & 5 & 3 & 0.0882 & 2 & 3 & 29 & 0.0 & 2 \\
 & 100 & 17 & 1 & 0.0294 & 34 & 34 & 34 & 0.0 & 12.412 \\
\addlinespace
\multirow{5}{*}{Cora} & 0 & 0 & 2365 & 0.8733 & 1 & 1 & 114 & 81.7 & 0 \\
 & 25 & 2 & 86 & 0.0318 & 1 & 2 & 1582 & 1.4 & 1.996 \\
 & 50 & 3 & 41 & 0.0151 & 1 & 3 & 1989 & 0.6 & 2.972 \\
 & 75 & 5 & 18 & 0.0066 & 1 & 3 & 2404 & 0.2 & 4.475 \\
 & 100 & 168 & 1 & 0.0004 & 2708 & 2708 & 2708 & 0.0 & 164.102 \\
\addlinespace
\multirow{5}{*}{CiteSeer} & 0 & 0 & 2090 & 0.6282 & 1 & 1 & 498 & 54.1 & 0 \\
 & 25 & 1 & 291 & 0.0875 & 1 & 1 & 1379 & 5.2 & 0.999 \\
 & 50 & 2 & 90 & 0.0271 & 1 & 2 & 2174 & 1.1 & 1.999 \\
 & 75 & 3 & 44 & 0.0132 & 1 & 2 & 2612 & 0.5 & 2.89 \\
 & 100 & 99 & 1 & 0.0003 & 3327 & 3327 & 3327 & 0.0 & 96.264 \\
\addlinespace
\multirow{5}{*}{PubMed} & 0 & 0 & 12998 & 0.6592 & 1 & 1 & 41 & 55.2 & 0 \\
 & 25 & 1 & 1629 & 0.0826 & 1 & 1 & 12374 & 6.0 & 1 \\
 & 50 & 2 & 594 & 0.0301 & 1 & 1 & 13995 & 1.7 & 1.999 \\
 & 75 & 4 & 206 & 0.0104 & 1 & 2 & 15545 & 0.5 & 3.966 \\
 & 100 & 171 & 1 & \(5.1\times 10^{-5}\) & 19717 & 19717 & 19717 & 0.0 & 166.504 \\
\bottomrule
\end{tabular}%
}
\end{table}

\begin{figure}[h]
    \centering
    \includegraphics[width=0.7\linewidth]{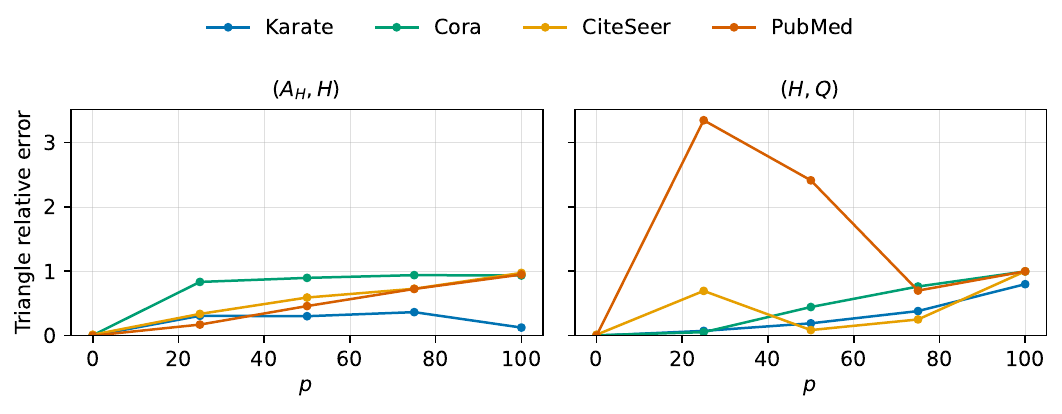}
    \caption{Triangle relative error across requested degree percentiles. Triangles are not constrained by the EP realizability problem and therefore need not be preserved by 1-WL-induced partitions.}
    \label{fig:utility-errors-percentiles}
\end{figure}

\begin{table}[h]
\centering
\caption{Original and reconstructed scalar metrics of Karate dataset. Avg.Rank is the average rank over the seven scalar errors; in bold we report the values below the average metric values over all methods}
\label{tab:scalar-errors-karate}
\scriptsize
\resizebox{\textwidth}{!}{%
\begin{tabular}{lr|rrrrrrr|r}
\toprule
Method & $p$ & Components & LCC & Clustering & Triangles & Wedges & Assortativity & Power-law & Avg.Rank \\
\midrule
Original & -- & 1 & 34 & 0.571 & 45 & 528 & -0.476 & 1.507 & -- \\
\midrule
\multirow{5}{*}{$(AH,H)$} & 0 & \textbf{1} & \textbf{34} & \textbf{0.571} & \textbf{45} & \textbf{528} & \textbf{-0.476} & \textbf{1.507} & \textbf{3.500} \\
 & 25 & \textbf{1} & \textbf{34} & 0.33 & 31 & \textbf{528} & \textbf{-0.476} & \textbf{1.507} & \textbf{6.429} \\
 & 50 & \textbf{1} & \textbf{34} & 0.337 & 31 & \textbf{528} & -0.472 & \textbf{1.507} & \textbf{6.571} \\
 & 75 & \textbf{1} & \textbf{34} & 0.307 & 29 & \textbf{528} & -0.472 & \textbf{1.507} & \textbf{7.429} \\
 & 100 & \textbf{1} & \textbf{34} & 0.353 & 40 & \textbf{528} & -0.295 & \textbf{1.507} & \textbf{6.571} \\
\addlinespace
\multirow{5}{*}{$(H,Q)$} & 0 & \textbf{1} & \textbf{34} & \textbf{0.571} & \textbf{45} & \textbf{528} & \textbf{-0.476} & \textbf{1.507} & \textbf{3.500} \\
 & 25 & 2 & 33 & 0.319 & 46 & 564 & -0.478 & 1.519 & 10.857 \\
 & 50 & \textbf{1} & \textbf{34} & 0.337 & 36 & 536 & -0.426 & 1.522 & 10.429 \\
 & 75 & \textbf{1} & \textbf{34} & 0.37 & 28 & 509 & -0.442 & 2.08 & 9.571 \\
 & 100 & \textbf{1} & \textbf{34} & 0.096 & 9 & 284 & -0.037 & 4.777 & 13.571 \\
\addlinespace
\multirow{5}{*}{$H$-SBM} & 0 & \textbf{1} & \textbf{34} & 0.589 & 46 & 530 & -0.48 & \textbf{1.507} & \textbf{6.286} \\
 & 25 & 2 & 33 & 0.304 & 37 & 555 & -0.474 & 1.516 & 11.143 \\
 & 50 & 2 & 33 & 0.297 & 40 & 569 & -0.426 & 1.52 & 12.286 \\
 & 75 & 2 & 33 & 0.275 & 32 & 540 & -0.432 & 1.545 & 12.286 \\
 & 100 & \textbf{1} & \textbf{34} & 0.131 & 16 & 347 & -0.064 & 1.629 & 14.429 \\
\addlinespace
Havel--Hakimi & -- & 4 & 24 & 0.721 & 88 & \textbf{528} & -0.045 & \textbf{1.507} & 11.429 \\
\addlinespace
Configuration & -- & \textbf{1} & \textbf{34} & 0.34 & 39 & \textbf{528} & -0.306 & \textbf{1.507} & \textbf{6.714} \\
\bottomrule
\end{tabular}%
}
\end{table}

\begin{table}[h]
\centering
\caption{Original and reconstructed scalar metrics of Cora dataset. Avg.Rank is the average rank over the seven scalar errors; in bold we report the values below the average metric values over all methods}
\label{tab:scalar-errors-cora}
\scriptsize
\resizebox{\textwidth}{!}{%
\begin{tabular}{lr|rrrrrrr|r}
\toprule
Method & $p$ & Components & LCC & Clustering & Triangles & Wedges & Assortativity & Power-law & Avg.Rank \\
\midrule
Original & -- & 78 & 2485 & 0.241 & 1630 & 52301 & -0.066 & 1.566 & -- \\
\midrule
\multirow{5}{*}{$(AH,H)$} & 0 & 74 & \textbf{2485} & \textbf{0.233} & 1620 & \textbf{52301} & \textbf{-0.066} & \textbf{1.566} & \textbf{2.571} \\
 & 25 & 34 & 2637 & 0.019 & 270 & \textbf{52301} & \textbf{-0.066} & \textbf{1.566} & \textbf{5.286} \\
 & 50 & 25 & 2657 & 0.012 & 167 & \textbf{52301} & \textbf{-0.066} & \textbf{1.566} & \textbf{6.143} \\
 & 75 & 19 & 2668 & 0.007 & 98 & \textbf{52301} & \textbf{-0.066} & \textbf{1.566} & \textbf{7.571} \\
 & 100 & 14 & 2680 & 0.007 & 107 & \textbf{52301} & -0.029 & \textbf{1.566} & \textbf{8.857} \\
\addlinespace
\multirow{5}{*}{$(H,Q)$} & 0 & 74 & \textbf{2485} & \textbf{0.233} & 1620 & \textbf{52301} & \textbf{-0.066} & \textbf{1.566} & \textbf{2.714} \\
 & 25 & 686 & 1716 & 0.008 & 1542 & 112367 & 0.137 & 1.631 & 12.857 \\
 & 50 & 443 & 1824 & 0.005 & 905 & 80797 & 0.215 & 1.675 & 13.571 \\
 & 75 & 92 & 2452 & 0.003 & 386 & 61623 & 0.138 & 1.606 & 10.143 \\
 & 100 & 1 & 2708 & \(8.3\times 10^{-4}\) & 4 & 15420 & 0.001 & 3.269 & 14.857 \\
\addlinespace
\multirow{5}{*}{$H$-SBM} & 0 & \textbf{80} & 2482 & \textbf{0.233} & \textbf{1626} & 52405 & -0.067 & 1.557 & \textbf{4.571} \\
 & 25 & 229 & 2467 & 0.008 & 118 & 56570 & -0.07 & 1.552 & 9.000 \\
 & 50 & 190 & 2508 & 0.006 & 92 & 56286 & -0.07 & 1.552 & 10.143 \\
 & 75 & 140 & 2562 & 0.006 & 69 & 54376 & -0.07 & 1.551 & 10.429 \\
 & 100 & 59 & 2648 & 0.001 & 10 & 20679 & \(-6.0\times 10^{-5}\) & 1.514 & 12.143 \\
\addlinespace
Havel--Hakimi & -- & 652 & 702 & 0.731 & 8172 & \textbf{52301} & 0.178 & \textbf{1.566} & 13.143 \\
\addlinespace
Configuration & -- & 13 & 2682 & 0.007 & 105 & \textbf{52301} & -0.03 & \textbf{1.566} & 9.000 \\
\bottomrule
\end{tabular}%
}
\end{table}

\begin{table}[h]
\centering
\caption{Original and reconstructed scalar metrics of CiteSeer dataset. Avg.Rank is the average rank over the seven scalar errors; in bold we report the values below the average metric values over all methods}
\label{tab:scalar-errors-citeseer}
\scriptsize
\resizebox{\textwidth}{!}{%
\begin{tabular}{lr|rrrrrrr|r}
\toprule
Method & $p$ & Components & LCC & Clustering & Triangles & Wedges & Assortativity & Power-law & Avg.Rank \\
\midrule
Original & -- & 438 & 2120 & 0.141 & 1167 & 26918 & 0.048 & 1.715 & -- \\
\midrule
\multirow{5}{*}{$(AH,H)$} & 0 & \textbf{423} & \textbf{2120} & 0.129 & 1150 & \textbf{26918} & \textbf{0.048} & \textbf{1.715} & \textbf{2.571} \\
 & 25 & 381 & 2482 & 0.038 & 773 & \textbf{26918} & \textbf{0.048} & \textbf{1.715} & \textbf{4.429} \\
 & 50 & 339 & 2557 & 0.015 & 476 & \textbf{26918} & \textbf{0.048} & \textbf{1.715} & \textbf{5.857} \\
 & 75 & 290 & 2677 & 0.007 & 317 & \textbf{26918} & \textbf{0.048} & \textbf{1.715} & \textbf{7.714} \\
 & 100 & 168 & 3008 & 0.002 & 30 & \textbf{26918} & -0.015 & \textbf{1.715} & 10.429 \\
\addlinespace
\multirow{5}{*}{$(H,Q)$} & 0 & \textbf{423} & \textbf{2120} & 0.129 & 1150 & \textbf{26918} & \textbf{0.048} & \textbf{1.715} & \textbf{2.714} \\
 & 25 & 1419 & 1636 & 0.034 & 1978 & 105609 & -0.072 & 1.639 & 12.000 \\
 & 50 & 1050 & 1737 & 0.01 & 1268 & 57566 & 0.221 & 1.724 & 10.714 \\
 & 75 & 775 & 1830 & 0.005 & 873 & 43221 & 0.385 & 1.797 & 11.143 \\
 & 100 & 1 & 3327 & \(3.7\times 10^{-4}\) & 1 & 8227 & \(-8.0\times 10^{-4}\) & 2.706 & 15.143 \\
\addlinespace
\multirow{5}{*}{$H$-SBM} & 0 & 461 & 2115 & \textbf{0.13} & \textbf{1156} & 27445 & 0.04 & 1.662 & \textbf{4.714} \\
 & 25 & 723 & 2484 & 0.021 & 563 & 30679 & 0.033 & 1.628 & 9.286 \\
 & 50 & 666 & 2565 & 0.008 & 330 & 31018 & 0.033 & 1.637 & 10.143 \\
 & 75 & 582 & 2663 & 0.004 & 181 & 30200 & 0.038 & 1.644 & 9.857 \\
 & 100 & 237 & 3060 & \(7.1\times 10^{-4}\) & 4 & 12458 & \(-7.1\times 10^{-4}\) & 1.618 & 13.143 \\
\addlinespace
Havel--Hakimi & -- & 978 & 1155 & 0.439 & 4877 & \textbf{26918} & 0.357 & \textbf{1.715} & 12.857 \\
\addlinespace
Configuration & -- & 171 & 3002 & 0.002 & 29 & \textbf{26918} & -0.016 & \textbf{1.715} & 10.286 \\
\bottomrule
\end{tabular}%
}
\end{table}

\begin{table}[h]
\centering
\caption{Original and reconstructed scalar metrics of PubMed dataset. Avg.Rank is the average rank over the seven scalar errors; in bold we report the values below the average metric values over all methods}
\label{tab:scalar-errors-pubmed}
\scriptsize
\resizebox{\textwidth}{!}{%
\begin{tabular}{lr|rrrrrrr|r}
\toprule
Method & $p$ & Components & LCC & Clustering & Triangles & Wedges & Assortativity & Power-law & Avg.Rank \\
\midrule
Original & -- & 1 & 19717 & 0.06 & 12520 & 699342 & -0.044 & 1.648 & -- \\
\midrule
\multirow{5}{*}{$(AH,H)$} & 0 & \textbf{1} & \textbf{19717} & \textbf{0.06} & \textbf{12520} & \textbf{699342} & \textbf{-0.044} & \textbf{1.648} & \textbf{2.500} \\
 & 25 & 25 & 19641 & 0.029 & 10366 & \textbf{699342} & \textbf{-0.044} & \textbf{1.648} & \textbf{4.286} \\
 & 50 & 35 & 19623 & 0.01 & 6764 & \textbf{699342} & \textbf{-0.044} & \textbf{1.648} & \textbf{5.286} \\
 & 75 & 63 & 19582 & 0.004 & 3435 & \textbf{699342} & \textbf{-0.044} & \textbf{1.648} & \textbf{6.571} \\
 & 100 & 508 & 18656 & 0.001 & 635 & \textbf{699342} & -0.005 & \textbf{1.648} & 9.143 \\
\addlinespace
\multirow{5}{*}{$(H,Q)$} & 0 & \textbf{1} & \textbf{19717} & \textbf{0.06} & \textbf{12520} & \textbf{699342} & \textbf{-0.044} & \textbf{1.648} & \textbf{2.500} \\
 & 25 & 15260 & 4267 & 0.017 & 54406 & 9995724 & -0.313 & 1.366 & 14.714 \\
 & 50 & 13932 & 5331 & 0.006 & 42723 & 3933211 & -0.189 & 1.412 & 14.143 \\
 & 75 & 11089 & 7701 & 0.002 & 21286 & 1598196 & 0.278 & 1.463 & 13.000 \\
 & 100 & \textbf{1} & \textbf{19717} & \(1.5\times 10^{-4}\) & 8 & 157422 & \(-4.9\times 10^{-4}\) & 5.095 & 11.143 \\
\addlinespace
\multirow{5}{*}{$H$-SBM} & 0 & 4 & 19714 & \textbf{0.06} & 12521 & 699349 & \textbf{-0.044} & \textbf{1.648} & \textbf{5.143} \\
 & 25 & 3615 & 16077 & 0.004 & 4501 & 734917 & -0.055 & 1.556 & 10.286 \\
 & 50 & 3286 & 16404 & 0.002 & 2175 & 736405 & -0.049 & 1.562 & 10.571 \\
 & 75 & 2742 & 16928 & 0.002 & 1361 & 729860 & -0.05 & 1.568 & 10.286 \\
 & 100 & 228 & 19485 & \(2.1\times 10^{-4}\) & 15 & 199121 & -0.001 & 1.479 & 11.714 \\
\addlinespace
Havel--Hakimi & -- & 6062 & 5684 & 0.394 & 115972 & \textbf{699342} & 0.773 & \textbf{1.648} & 12.714 \\
\addlinespace
Configuration & -- & 505 & 18663 & 0.001 & 644 & \textbf{699342} & -0.005 & \textbf{1.648} & 9.000 \\
\bottomrule
\end{tabular}%
}
\end{table}

\begin{figure}[h]
\centering
\includegraphics[width=\textwidth]{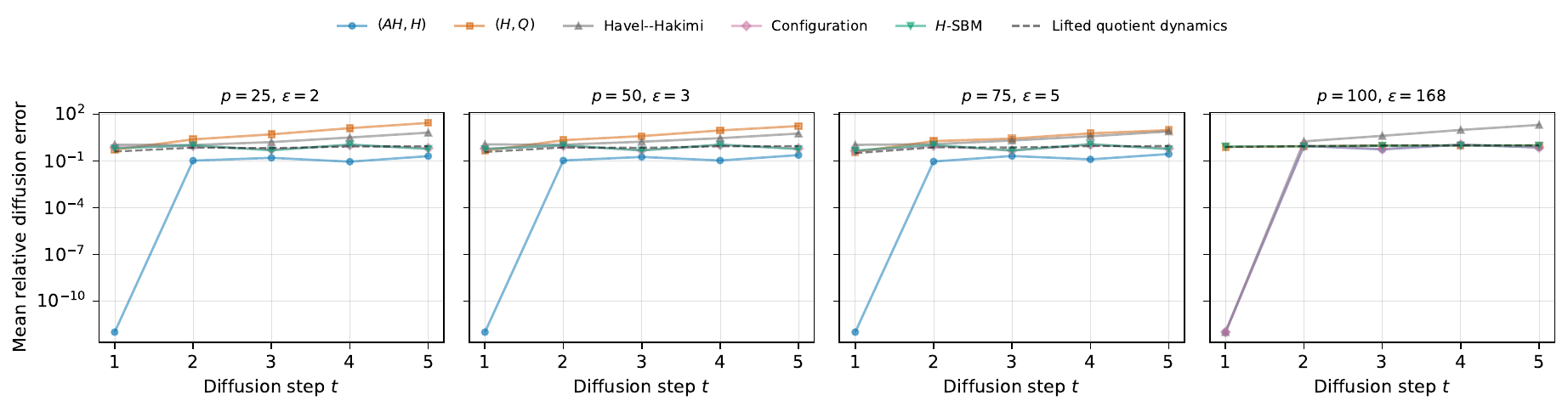}
\caption{Mean relative diffusion errors on Cora for $t=1,\ldots,5$.}
\label{fig:diffusion-errors-cora}
\end{figure}

\begin{figure}[h]
\centering
\includegraphics[width=\textwidth]{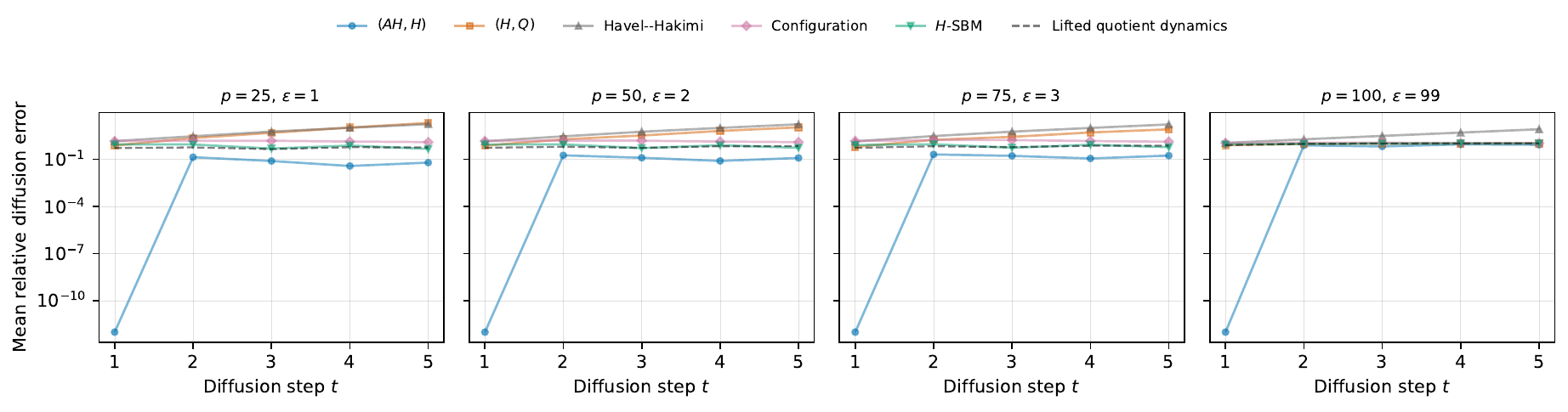}
\caption{Mean relative diffusion errors on Citeseer for $t=1,\ldots,5$.}
\label{fig:diffusion-errors-citeseer}
\end{figure}

\begin{figure}[h]
\centering
\includegraphics[width=\textwidth]{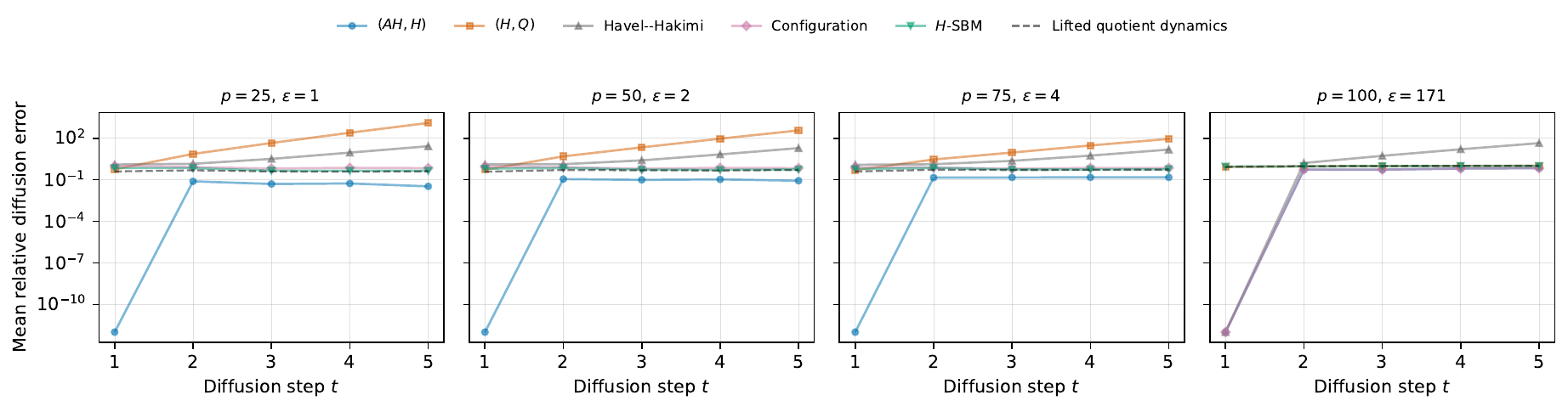}
\caption{Mean relative diffusion errors on PubMed for $t=1,\ldots,5$.}
\label{fig:diffusion-errors-pubmed}
\end{figure}

\begin{figure}[h]
\centering
\includegraphics[width=\textwidth]{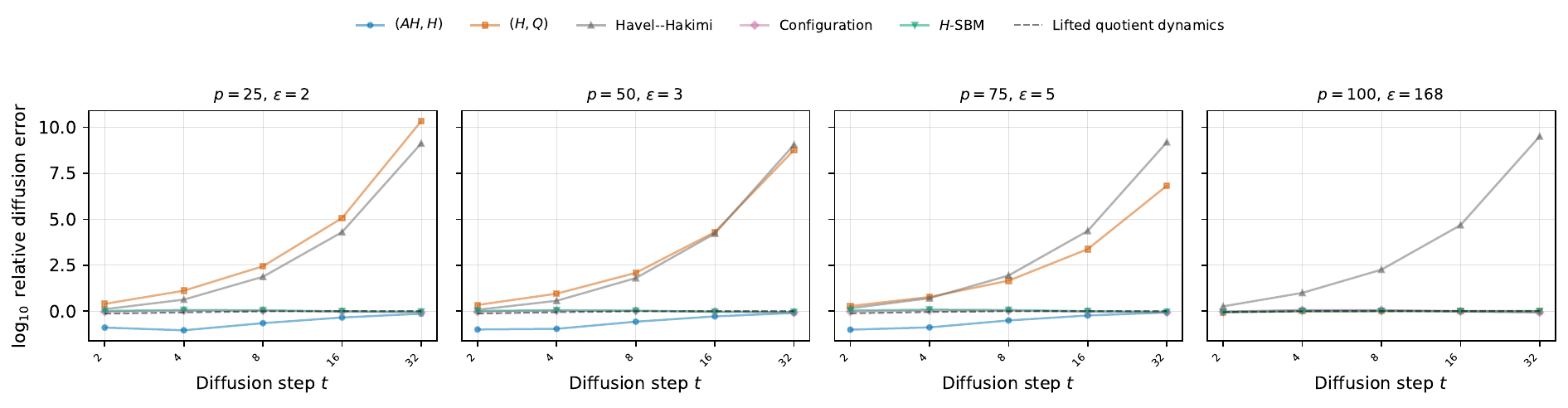}
\caption{Mean relative diffusion errors on Cora for $t\in\{2,4,8,16,32\}$. The x-axis is logarithmic with base 2 and the y-axis reports $\log_{10}$ relative error.}
\label{fig:diffusion-long-horizon-cora}
\end{figure}

\begin{figure}[h]
\centering
\includegraphics[width=\textwidth]{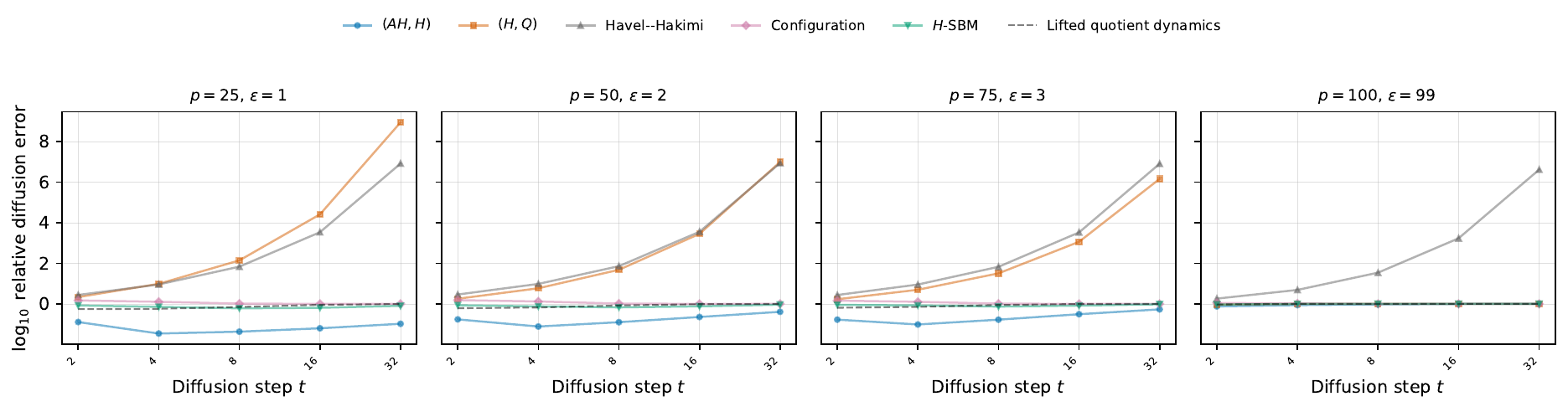}
\caption{Mean relative diffusion errors on Citeseer for $t\in\{2,4,8,16,32\}$. The x-axis is logarithmic with base 2 and the y-axis reports $\log_{10}$ relative error.}
\label{fig:diffusion-long-horizon-citeseer}
\end{figure}

\begin{figure}[h]
\centering
\includegraphics[width=\textwidth]{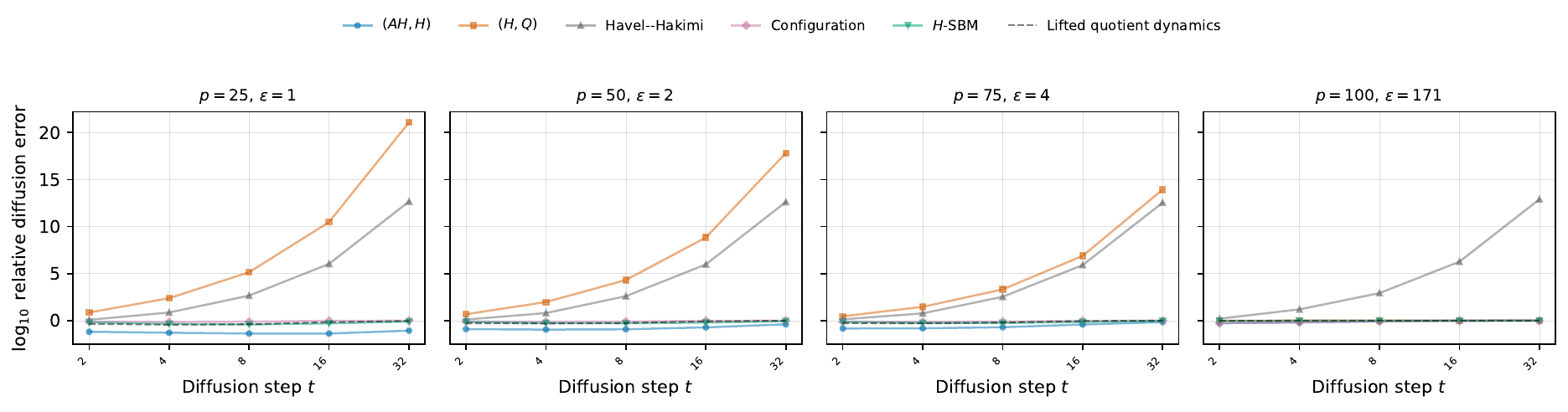}
\caption{Mean relative diffusion errors on PubMed for $t\in\{2,4,8,16,32\}$. The x-axis is logarithmic with base 2 and the y-axis reports $\log_{10}$ relative error.}
\label{fig:diffusion-long-horizon-pubmed}
\end{figure}

\begin{figure}[h]
\centering

\includegraphics[width=\linewidth]{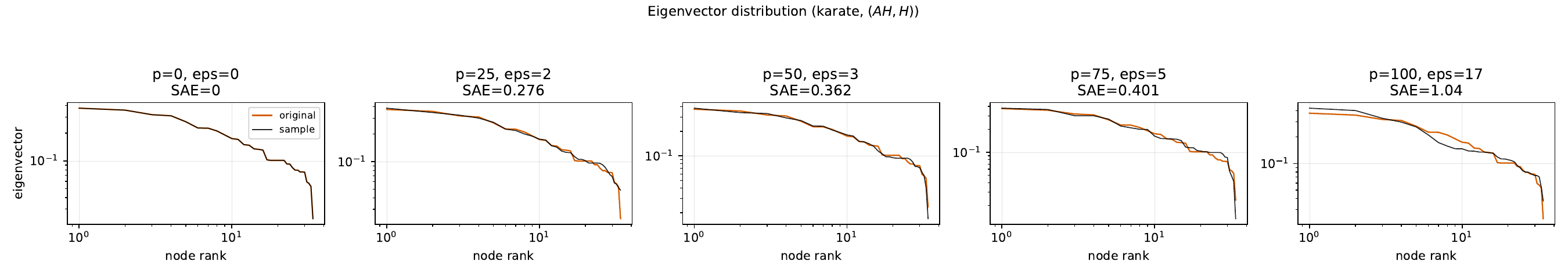}

\includegraphics[width=\linewidth]{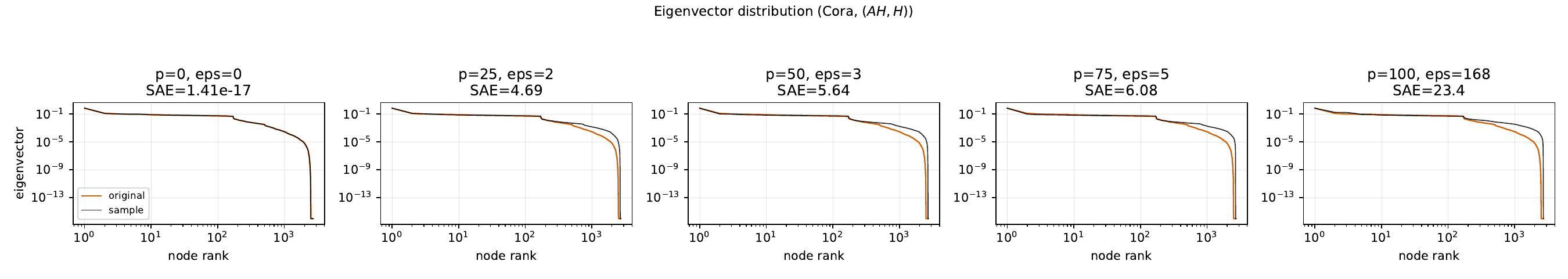}

\includegraphics[width=\linewidth]{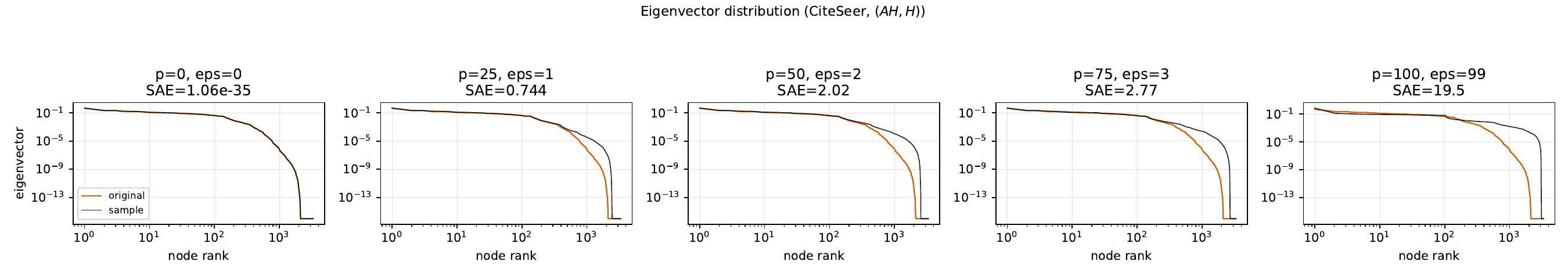}

\includegraphics[width=\linewidth]{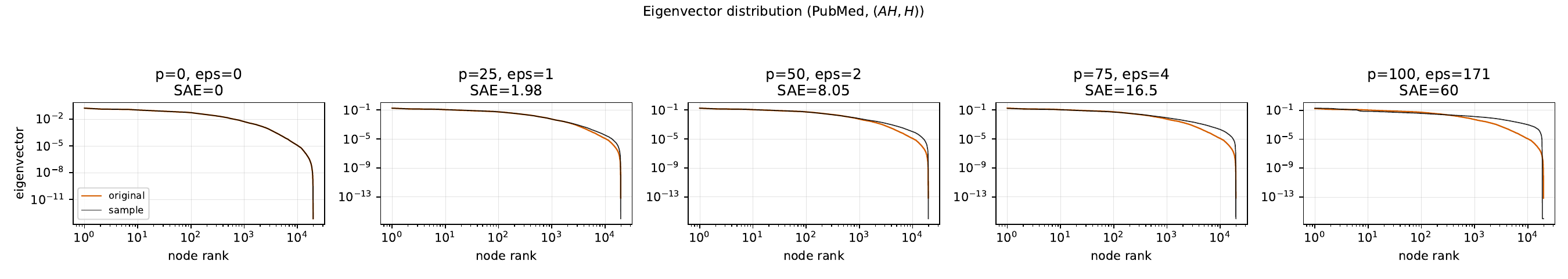}

\caption{Eigenvector centrality distributions of $\varepsilon$-EP reconstruction with release $(AH,H)$ and original graphs.}
\label{fig:eigenvector-distributions-ahh}
\end{figure}

\begin{figure}[h]
\centering

\includegraphics[width=\linewidth]{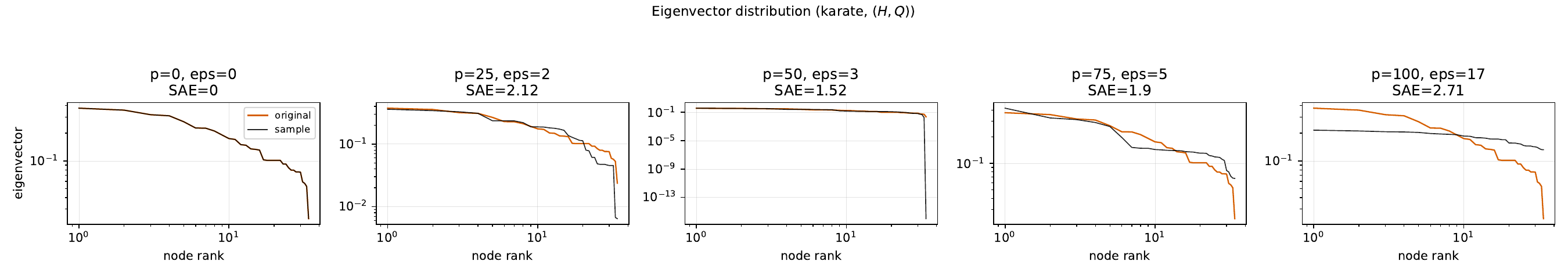}

\includegraphics[width=\linewidth]{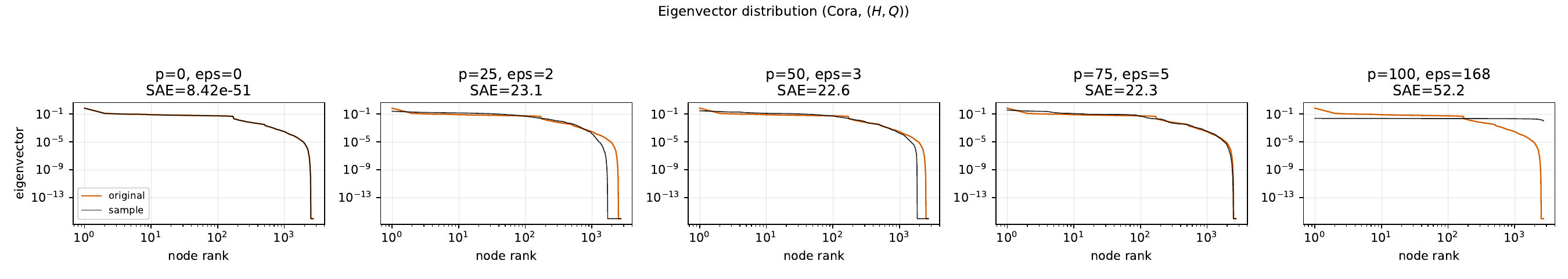}

\includegraphics[width=\linewidth]{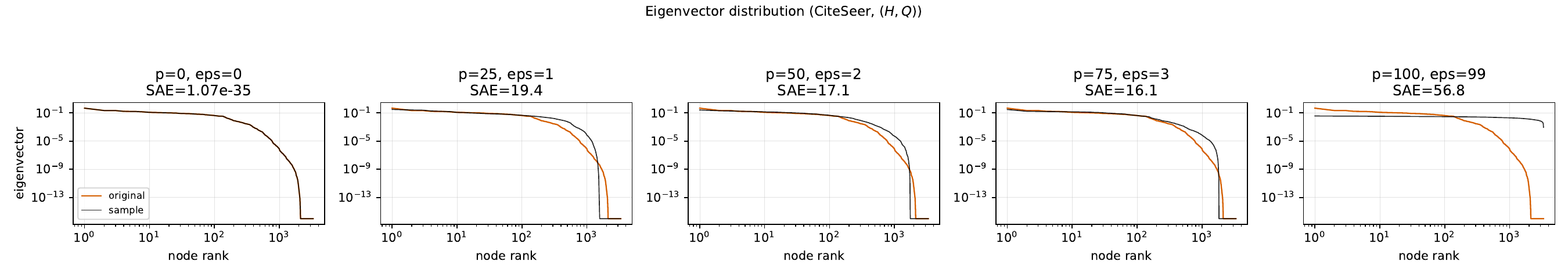}

\includegraphics[width=\linewidth]{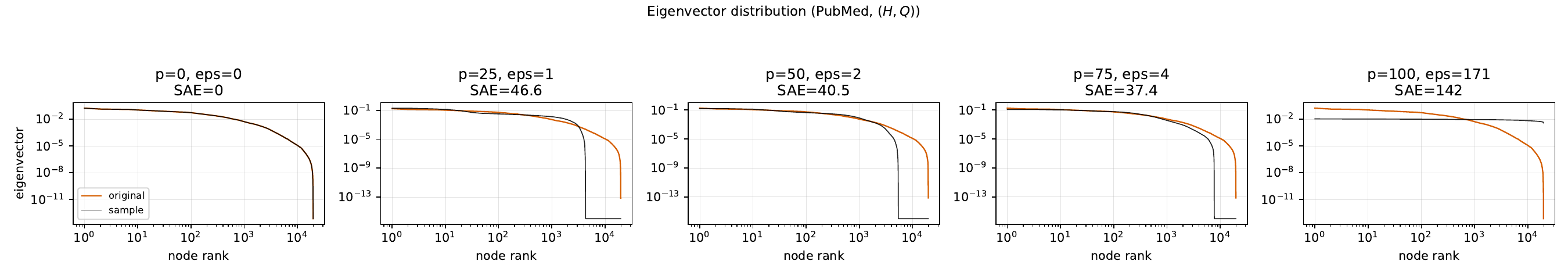}

\caption{Eigenvector centrality distributions of $\varepsilon$-EP reconstruction with release $(H,Q)$ and original graphs.}
\label{fig:eigenvector-distributions-hq}
\end{figure}

\begin{figure}[h]
\centering

\includegraphics[width=\linewidth]{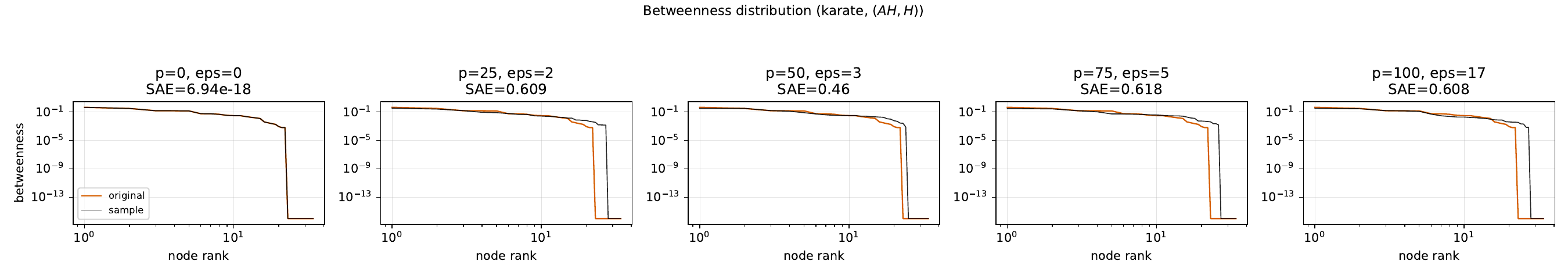}

\includegraphics[width=\linewidth]{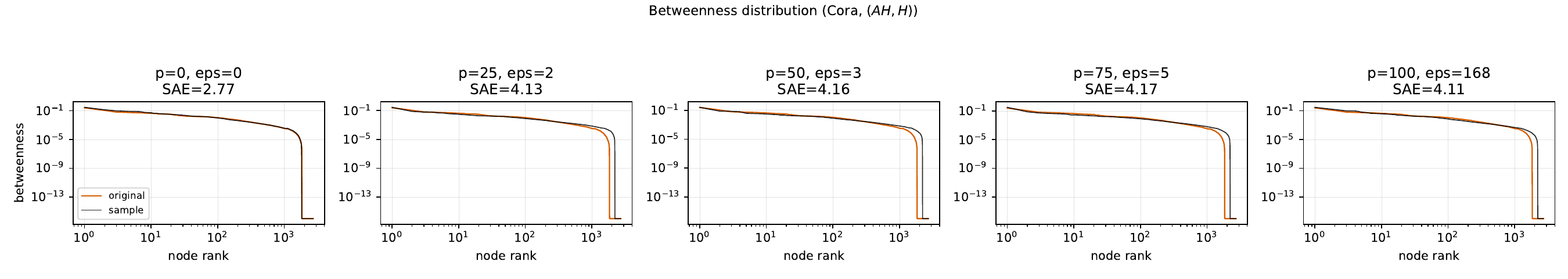}

\includegraphics[width=\linewidth]{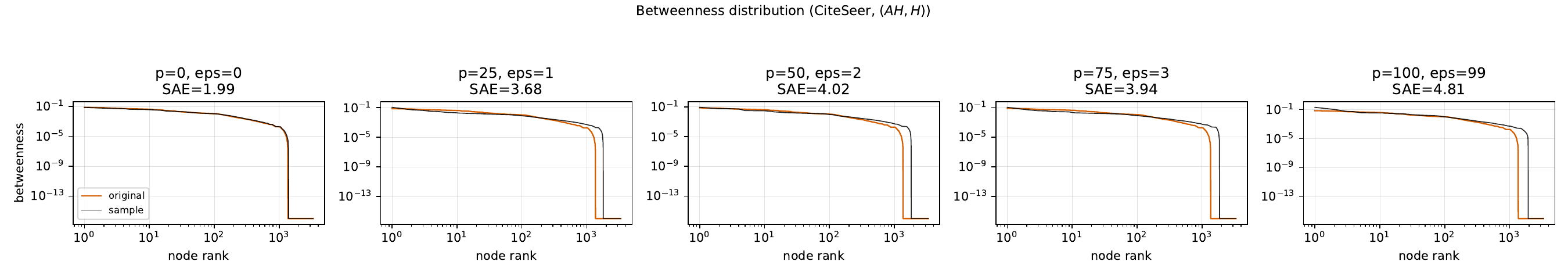}

\includegraphics[width=\linewidth]{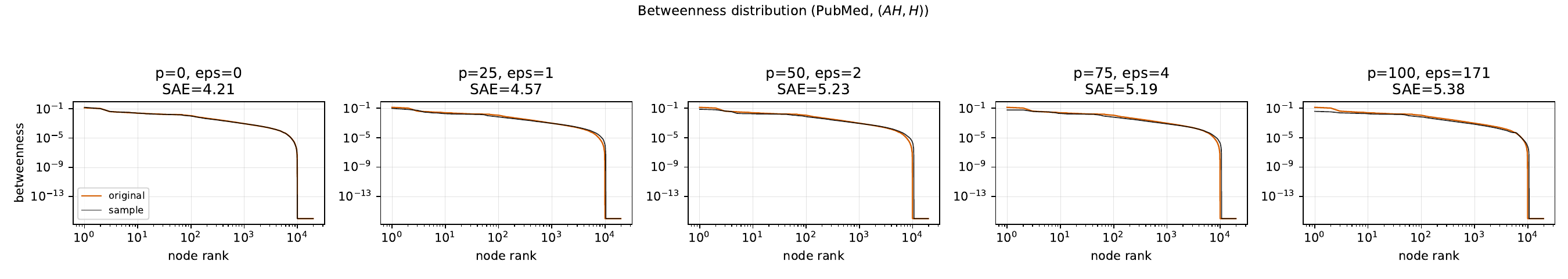}

\caption{Betweenness centrality distributions of $\varepsilon$-EP reconstruction with release $(AH,H)$ and original graphs.}
\label{fig:betweenness-distributions-ahh}
\end{figure}

\begin{figure}[h]
\centering

\includegraphics[width=\linewidth]{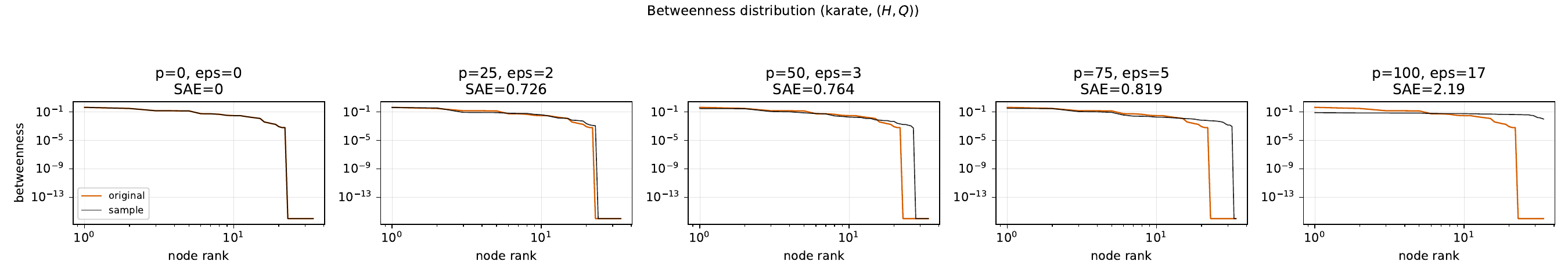}

\includegraphics[width=\linewidth]{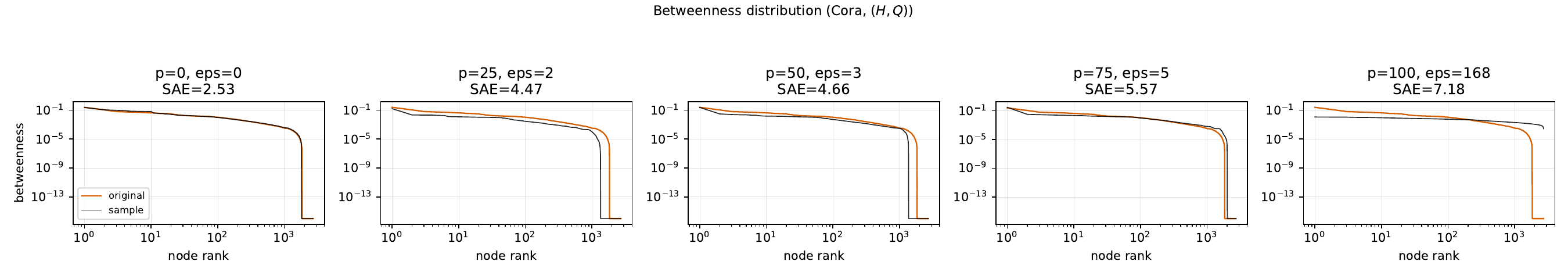}

\includegraphics[width=\linewidth]{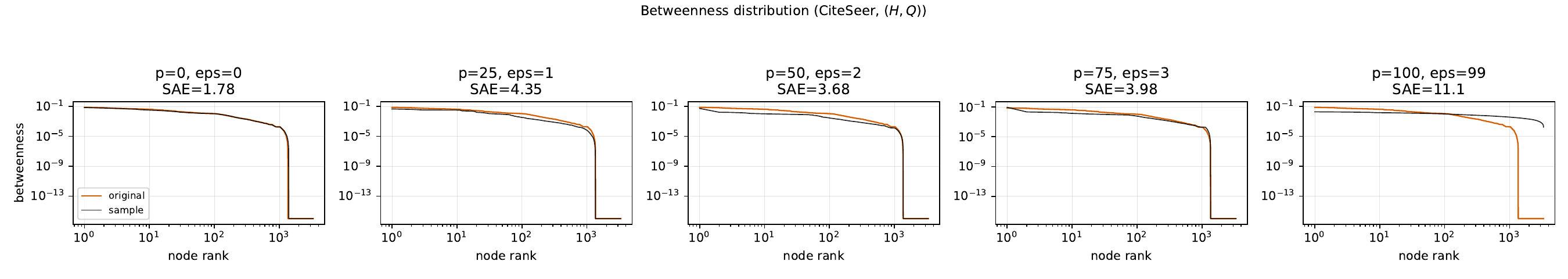}

\includegraphics[width=\linewidth]{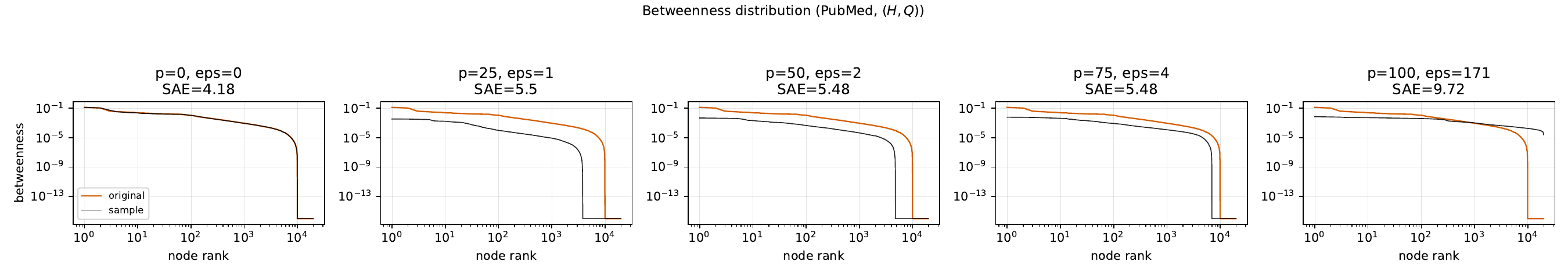}

\caption{Betweenness centrality distributions of $\varepsilon$-EP reconstruction with release $(H,Q)$ and original graphs.}
\label{fig:betweenness-distributions-hq}
\end{figure}

\begin{figure}[h]
\centering

\includegraphics[width=\linewidth]{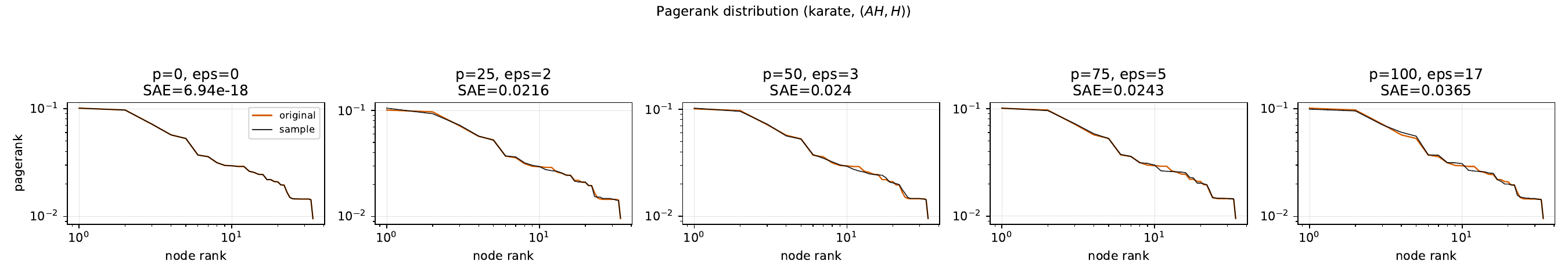}

\includegraphics[width=\linewidth]{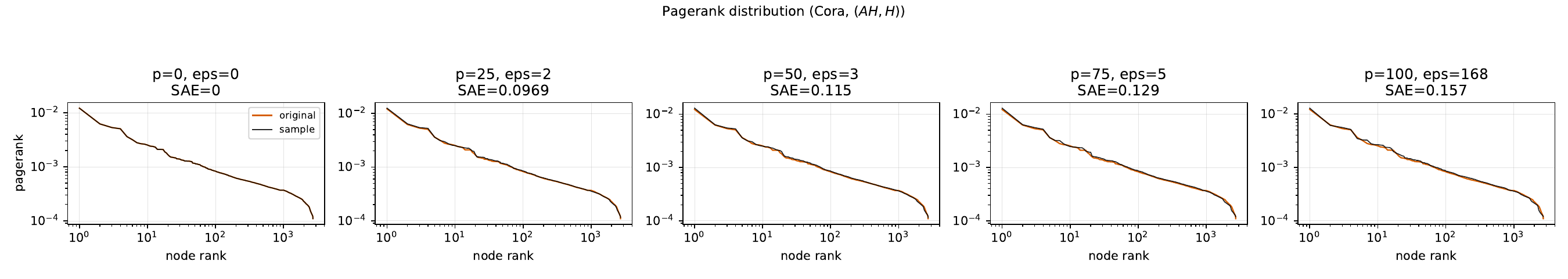}

\includegraphics[width=\linewidth]{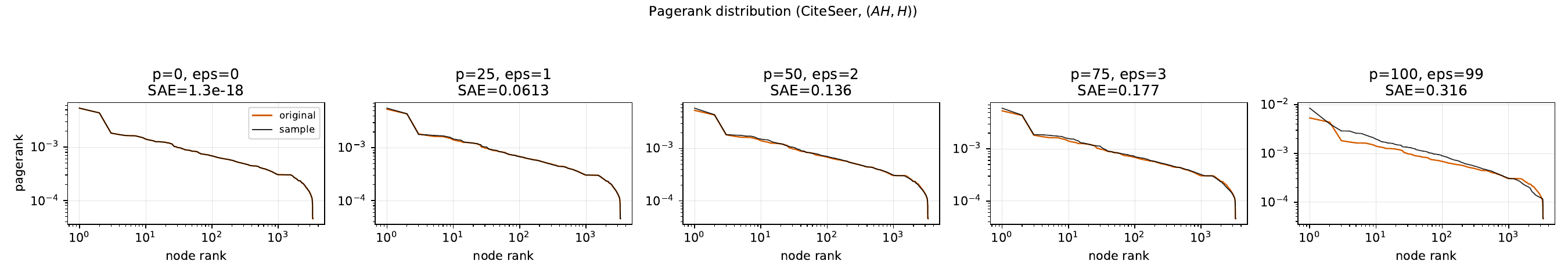}

\includegraphics[width=\linewidth]
{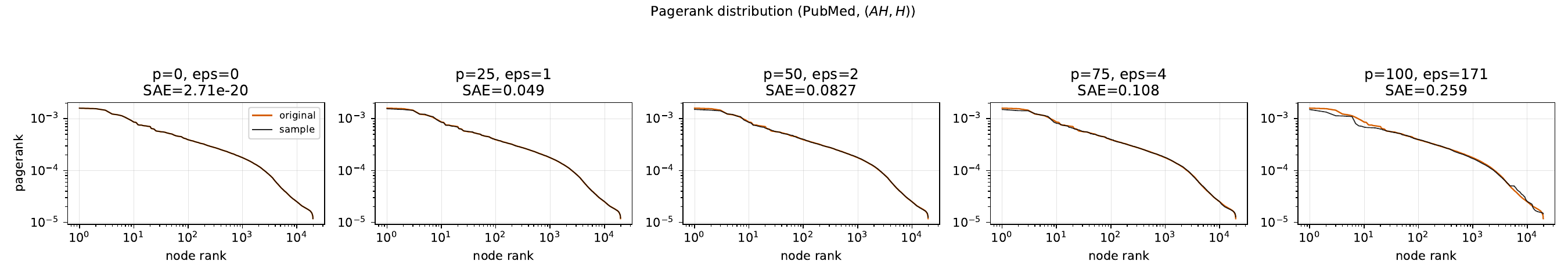}

\caption{PageRank centrality distributions of $\varepsilon$-EP reconstruction with release $(AH,H)$ and original graphs.}
\label{fig:pagerank-distributions-ahh}
\end{figure}

\begin{figure}[h]
\centering

\includegraphics[width=\linewidth]{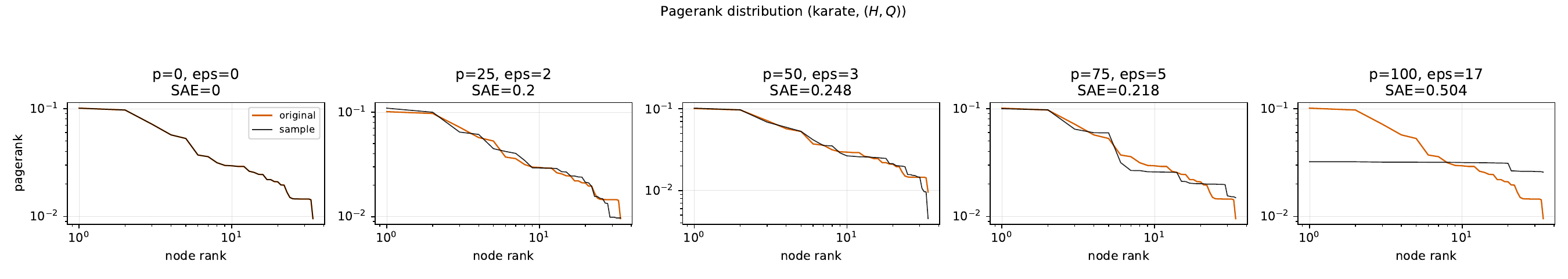}

\includegraphics[width=\linewidth]{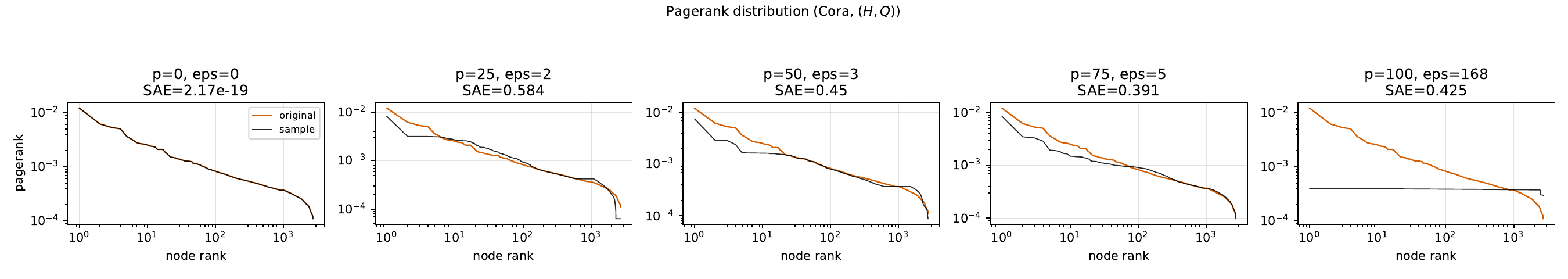}

\includegraphics[width=\linewidth]{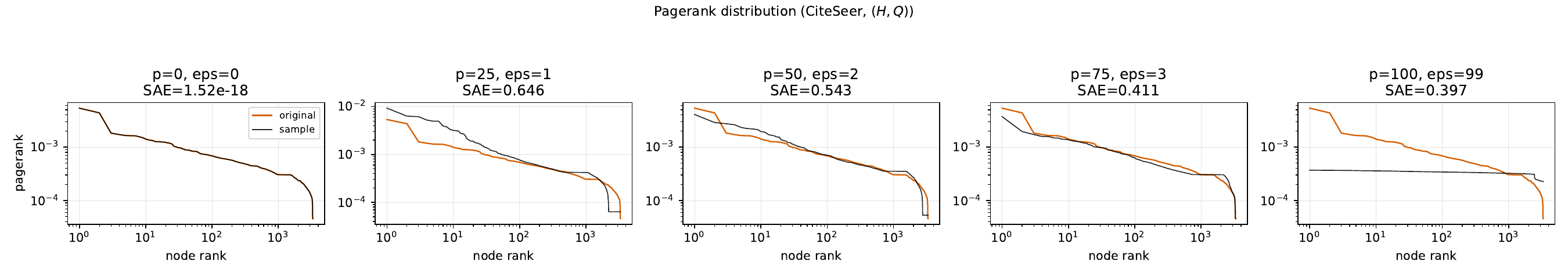}

\includegraphics[width=\linewidth]{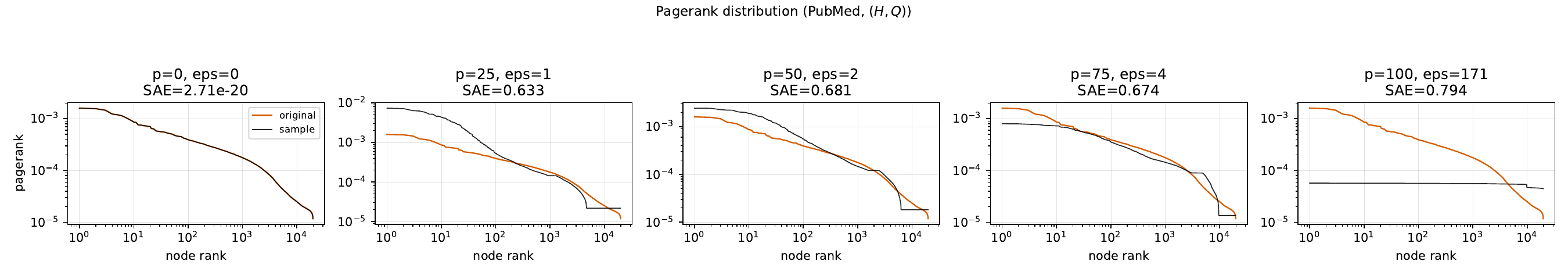}

\caption{PageRank centrality distributions of $\varepsilon$-EP reconstruction with release $(H,Q)$ and original graphs.}
\label{fig:pagerank-distributions-hq}
\end{figure}

\begin{figure}[h]
\centering

\includegraphics[width=\linewidth]{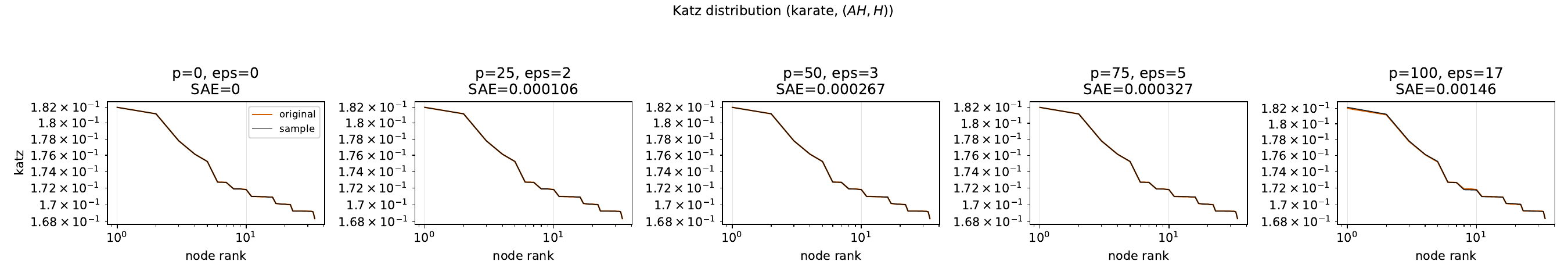}

\includegraphics[width=\linewidth]{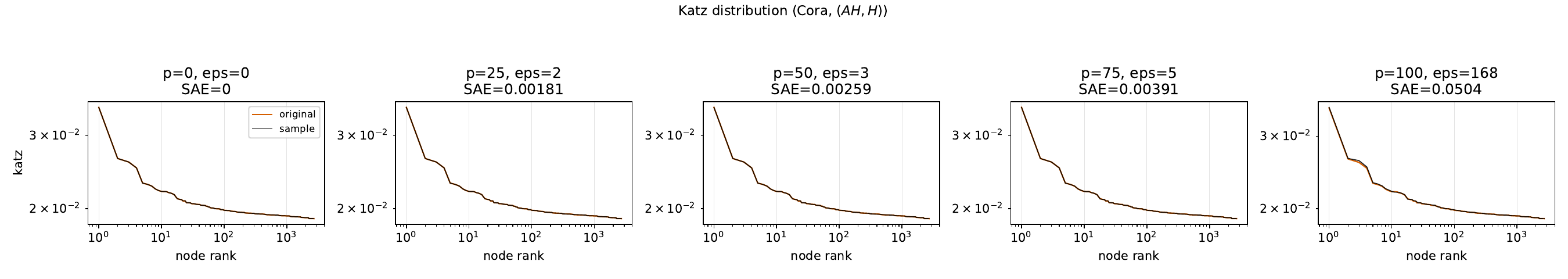}

\includegraphics[width=\linewidth]{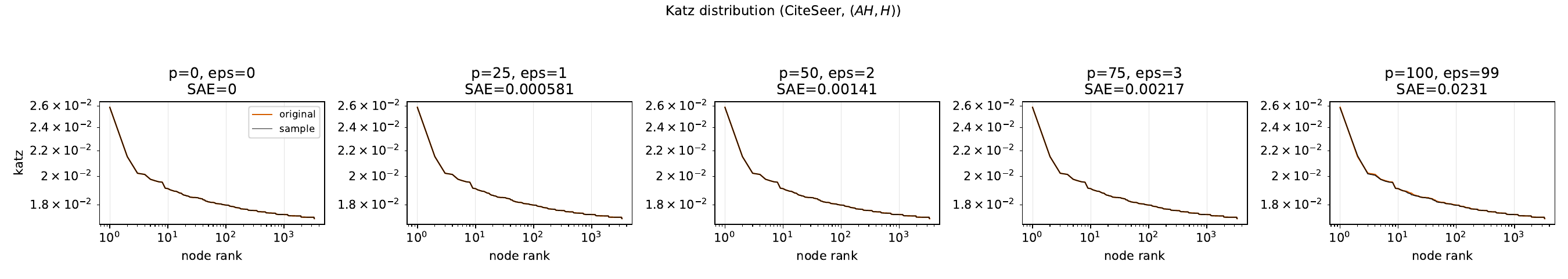}

\includegraphics[width=\linewidth]{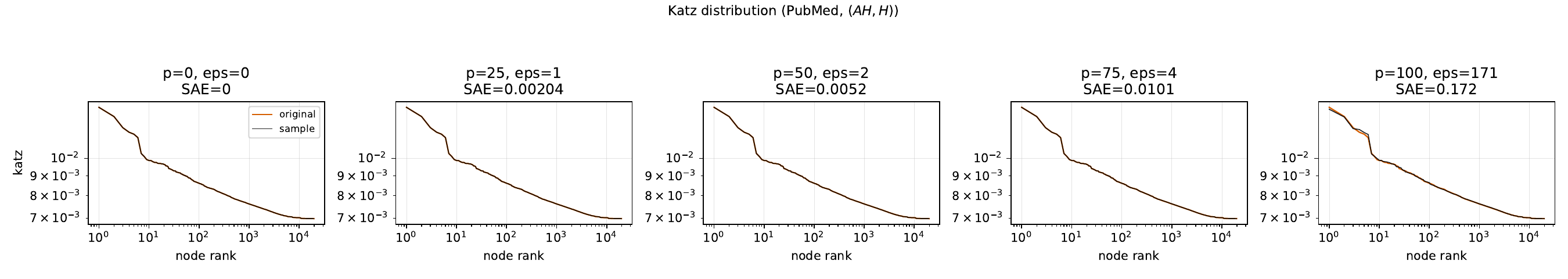}

\caption{Katz centrality distributions of $\varepsilon$-EP reconstruction with release $(AH,H)$ and original graphs.}
\label{fig:katz-distributions-ahh}
\end{figure}

\begin{figure}[h]
\centering
\includegraphics[width=\linewidth]{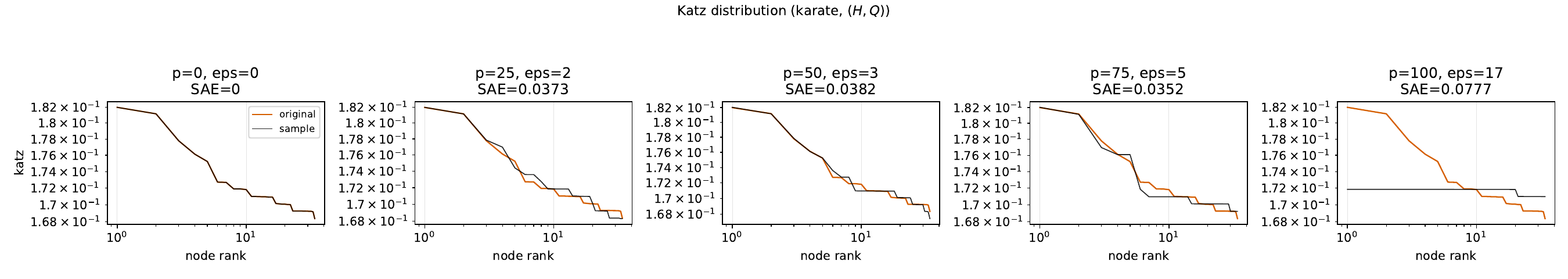}

\includegraphics[width=\linewidth]{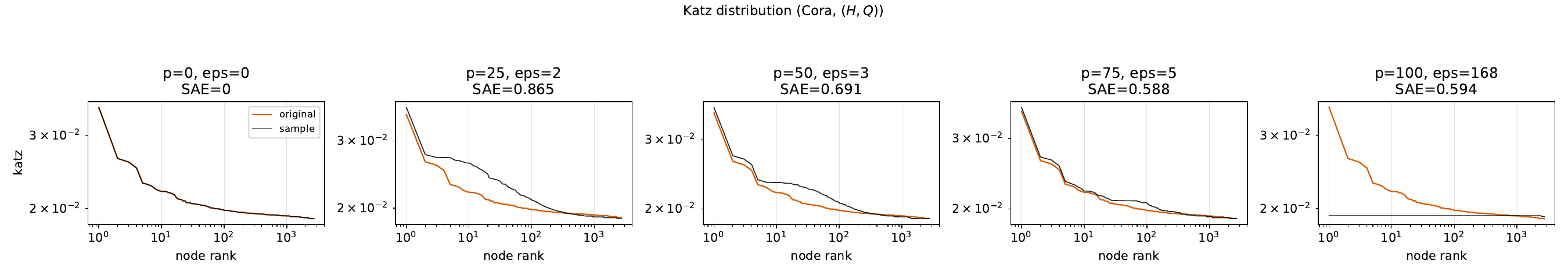}

\includegraphics[width=\linewidth]{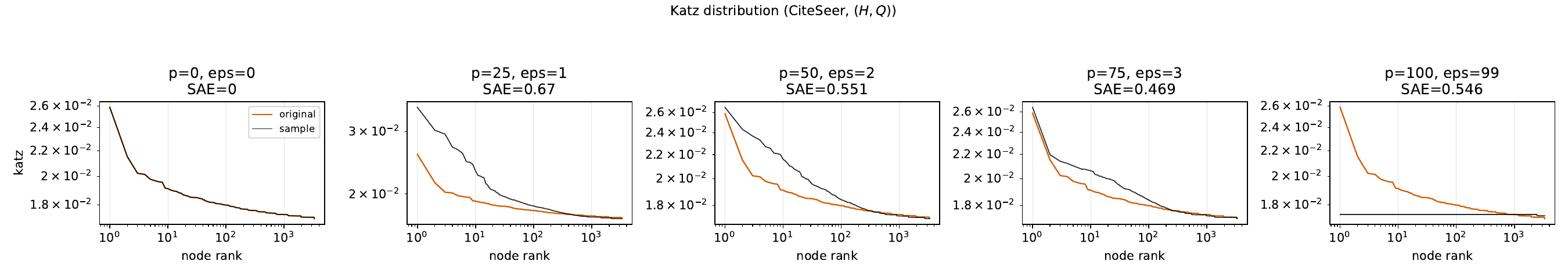}

\includegraphics[width=\linewidth]{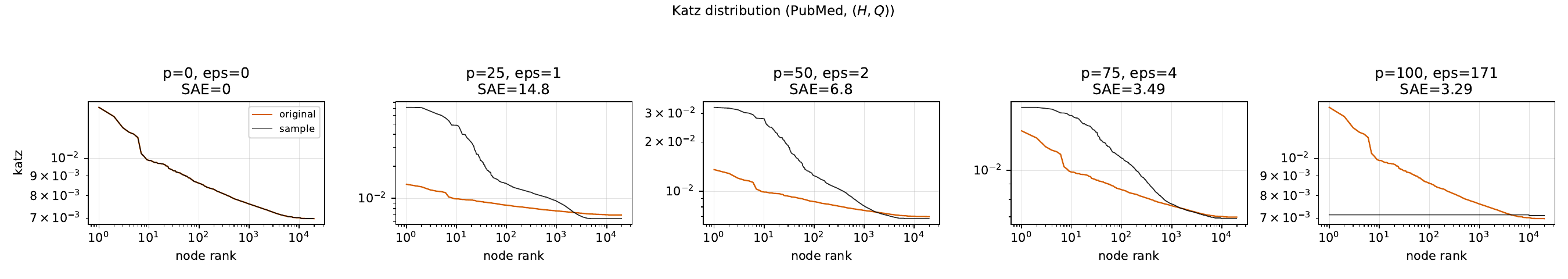}

\caption{Katz centrality distributions of $\varepsilon$-EP reconstruction with release $(H,Q)$ and original graphs.}
\label{fig:katz-distributions-hq}
\end{figure}

\clearpage

\bibliographystyle{cas-model2-names}
\bibliography{ep_realizability_final}